\documentclass[12pt]{article}
\usepackage[bookmarks={true},bookmarksopen={true},colorlinks,citecolor = blue,anchorcolor = blue,linkcolor = blue]{hyperref}
\usepackage{xcolor}
\usepackage[T1]{fontenc}
\usepackage[utf8]{inputenc}
\usepackage[english]{babel}
\usepackage{lmodern}
\usepackage{array}
\usepackage{amsmath}
\usepackage{amssymb}
\usepackage{amsfonts}
\usepackage{amsthm}
\usepackage{mathrsfs}
\usepackage{enumitem}
\usepackage{pgf}
\usepackage{bbm}
\usepackage{natbib}
\usepackage[ruled,vlined]{algorithm2e}
\usepackage{float}
\usepackage{subfig}
\usepackage{caption}
\usepackage{xr}

\newcommand{\blind}{1}

\numberwithin{equation}{section}

\renewcommand{\leq}{\leqslant}
\renewcommand{\geq}{\geqslant}

\renewcommand{\epsilon}{\varepsilon}
\renewcommand{\ln}{\log}

\def\twofig{6.85cm}
\def\Vplot{5cm}
\def\densplot{5cm}

\newcommand{\E}{\mathbb{E}}

\newcommand{\N}{\mathbb{N}}
\newcommand{\R}{\mathbb{R}}

\newcommand{\Z}{\mathbb{Z}}

\newcommand{\der}{\mathrm{d}}

\newcommand{\cara}[1]{\mathbbm{1}_{#1}}

\newcommand{\defeq}{:=}

\newcommand{\diag}{\mbox{diag}}

\newcommand{\var}{\mbox{Var}}

\newcommand{\tr}{\mathrm{Tr}}

\newcommand{\Xnew}{X_{\mathrm{new}}}
\newcommand{\Xbest}{X_{\mathrm{best}}}
\newcommand{\pcov}{\Lambda}

\newcommand{\PVS}[3]{\mathbb{P}_{\mathrm{VS}}^{#1}(#2,#3)}

\DeclareMathOperator*{\argmin}{arg\,min}
\DeclareMathOperator*{\argmax}{arg\,max}

\theoremstyle{plain}
\newtheorem{prop}{Proposition}[section]

\newtheorem{defin}[prop]{Definition}
\newtheorem{defin/prop}[prop]{Definition/Proposition}

\newtheorem{corr}[prop]{Corollary}

\title{On proportional volume sampling for experimental design in general spaces}
\author{Arnaud Poinas and R\'emi Bardenet}
\date{}
\begin{document}

\def\spacingset#1{\renewcommand{\baselinestretch}%
{#1}\small\normalsize} \spacingset{1}


\if1\blind
{
  \title{\bf On proportional volume sampling for experimental design in general spaces}
  \author{Arnaud Poinas\thanks{Corresponding author}~ and R\'emi Bardenet\\
    Universit\'e de Lille, CNRS, Centrale Lille,\\ UMR 9189 – CRIStAL,
    59651 Villeneuve d’Ascq, France}
  \maketitle
} \fi

\if0\blind
{
  \bigskip
  \bigskip
  \bigskip
  \begin{center}
    {\LARGE\bf On proportional volume sampling for experimental design in general spaces}
\end{center}
  \medskip
} \fi

\bigskip
\begin{abstract}
Optimal design for linear regression is a fundamental task in statistics.
For finite design spaces, recent progress has shown that random designs drawn using \emph{proportional volume sampling} (PVS) lead to approximation guarantees for A-optimal design.
PVS strikes the balance between design nodes that jointly fill the design space, while marginally staying in regions of high mass under the solution of a relaxed convex version of the original problem.
In this paper, we examine some of the statistical implications of a new variant of PVS for (possibly Bayesian) optimal design.
Using point process machinery, we treat the case of a generic Polish design space.
We show that not only are the A-optimality approximation guarantees preserved, but we obtain similar guarantees for D-optimal design that tighten recent results.
Moreover, we show that PVS can be sampled in polynomial time.
Unfortunately, in spite of its elegance and tractability, we demonstrate on a simple example that the practical implications of general PVS are likely limited.
In the second part of the paper, we focus on applications and investigate the use of PVS as a subroutine for stochastic search heuristics.
We demonstrate that PVS is a robust addition to the practitioner's toolbox, especially when the regression functions are nonstandard and the design space, while low-dimensional, has a complicated shape (e.g., nonlinear boundaries, several connected components).
\end{abstract}

\noindent%
{\it Keywords:} Bayesian optimal design; volume sampling; determinantal point processes.
\vfill

\newpage
\spacingset{1.5} 

\section{Introduction}
In the classical problem of experimental design for linear models, the main goal is to select the input variables so as to minimize the variance of the Gauss-Markov estimator.
The literature and the available techniques are usually partitioned according to which set $\Omega$ of possible input variables is considered, the so-called \emph{design space}.
For instance, taking $\Omega$ to be a finite set
leads to factorial experiments \citep*[Chapter 7]{Atkinson} and row subset selection \citep*{Def_RescaledVS}.
Very often, $\Omega$ is a rectangular subset of $\R^d$, but more complicated shapes are also not uncommon in applications. For example, $\Omega$ becomes a simplex when considering mixture designs \citep[Chapter 16]{Atkinson}, and some of the variables can satisfy nonlinear constraints for physical reasons \citep[Example 16.3]{Space_Filling_Design, Atkinson}.
Finally, taking  $\Omega$ to be a product of both a finite set and a compact subset of $\R^d$ allows including both qualitative and quantitative variables \cite[Example 1.2]{Atkinson}.
In this paper, we consider the general setting where the design space $\Omega$ is any bounded closed Polish space $\Omega$, thus including all previous examples.

In general, the problem of minimizing the variance matrix of the Gauss-Markov estimator is difficult to solve, both theoretically and numerically.
For starters, the Loewner order on positive definite symmetric matrices is not total.
The common workaround is to minimize instead a real-valued function of these matrices, which somehow measures the size of the variance of the Gauss-Markov estimator.
Even using such proxies, the resulting optimization problem remains difficult to tackle. For finite design spaces, for instance, one faces a combinatorial optimization problem, typically framed as an optimization over integer-weighted measures on $\Omega$. One seminal idea in this discrete setting has been to rather solve a continuous relaxation of the original minimization problem, yielding a weighted measure on $\Omega$ with \emph{real} weights, called an \emph{approximate} optimal design.
Several postprocessing procedures have been investigated to recover an integer-weighted measure from that real-weighted approximate optimal design \citep[Chapter 12]{Pukelsheim}.
Recently, \cite*{Nikolov} showed how to use that approximate design to define \emph{proportional volume sampling}, a distribution over subsets of a (large) finite set that $(i)$ typically charges designs closer to optimality than i.i.d. random designs, and $(ii)$ can be sampled in polynomial time.  The same distribution was also introduced in \citep*{Def_RescaledVS} for a similar purpose, and generalized to a continuous design space in \citep*{VRS}. The key idea is to sample designs that strike a balance between, on the one side, each design point having a large mass under the approximate design, and on the other side, the points together corresponding to a small confidence ellipsoid for the Gauss-Markov estimator.
Finally, \cite*{RegDPP} modified proportional volume sampling distribution so that it applies to Bayesian optimal design on a finite design space.

In this paper, we investigate two natural ways of using proportional volume sampling (PVS) to generate random designs close to optimality in \emph{general} design spaces, in both the Bayesian and classical setting, thus extending the results of \citep{VRS,RegDPP}.
After a short survey of optimal design and related work in Section \ref{sec:Rel_Works}, we show in Section \ref{sec:GPCVS} that PVS has a natural extension to $\Omega$ being any Polish set.
We show that this extension preserves important known properties of PVS, i.e., that it yields unbiased estimates of the inverse information matrix and its determinant.
We also prove new approximation guarantees for A-optimal and D-optimal designs generated from PVS.
Because we focus on the standard statistical formulation of experimental design \citep{Atkinson}, our guarantees use PVS conditioned on the design having a fixed, used-defined cardinality.
Our results apply to any setting, Bayesian or not, finite or not, and slightly improve the known bounds in the finite case of \cite{RegDPP}.
Furthermore, we give a sampling algorithm for random designs generated from this general conditioned PVS. This algorithm further highlights the connection between PVS and determinantal point processes established in the finite frequentist setting by \cite{Nikolov}.
Using numerical simulations, we confirm that general PVS generates better Bayesian experimental designs than random i.i.d. designs.
However, in spite of its mathematical and algorithmic elegance, we found general PVS to bring only modest improvements in practical problems, compared to a second way of using finite PVS.
We describe this alternative extension of \citep{Nikolov,RegDPP} in Section~\ref{sec:algo}, namely a global search heuristic that generates designs close to A and D-optimality when $\Omega$ is any compact set of $\R^d$.
This heuristic uses finite PVS as a subroutine.
While standard algorithms, including exact methods when applicable \citep{DeCastro}, remain the preferred solution in the typical setting of a rectangular design space and low-degree polynomial regression functions, we show on several examples in Section~\ref{sec:perf} that our PVS-based heuristic outperforms them when the design space has a more complicated shape and the regression functions are generic.

\section{Background and related work}\label{sec:Rel_Works}
In this section, we recall the usual mathematical setting for optimal design in linear regression, and survey a few key results in the rich statistical literature on the subject.

\subsection{Optimal design for linear models}
Let $y_1,\cdots, y_k\in\R$ denote the responses of a fixed number $k\geq p$ of independent experiments with input variables $(x_1,\cdots,x_k)\in\Omega^k$. Consider a linear regression model over $p$ linearly independent regression functions $\phi_1,\cdots,\phi_p\in L^2(\Omega)$,
\begin{equation}\label{linreg}
Y=\phi(X)\beta+\epsilon,
\end{equation}
where $X\defeq (x_1,\cdots,x_k)^T$, $Y\defeq (y_1,\cdots,y_k)^T$, $\epsilon\in\R^k$ is a random variable such that $\E[\epsilon]=0$ and $\var(\epsilon)=\sigma^2 I_k$, and $\phi(X) = (\phi_j(x_i)) \in\R^{k\times p}$
is called the \emph{design matrix}.
The feature functions $\phi_i$ are often polynomials, but experimental design for different regression functions has also been considered, such as trigonometric polynomials \citep{Trigo_regression}, Haar wavelets \citep{Wavelet_reg} or B-splines \citep{B-Spline_1}.
Assuming that the information matrix $\phi(X)^T\phi(X)+\pcov$ is invertible, one popular estimator for $\beta$ in \eqref{linreg} is
\begin{equation} \label{eq:GM_estimator}
\hat\beta\defeq (\phi(X)^T\phi(X)+\pcov)^{-1}\phi(X)^T Y,
\end{equation}
where $\Lambda$ be a positive semi-definite matrix. When $\Lambda=0$, \eqref{eq:GM_estimator} corresponds to the classical Gauss-Markov estimator, while the case $\Lambda\neq 0$ corresponds to the posterior mean of $\beta$ for any prior distribution with mean $0$ and variance $\sigma^2 \pcov^{-1}$.
In the rest of the paper, we often abusively refer to $\Lambda\neq 0$ as the Bayesian setting, although frequentists may also prefer to use a nonzero $\Lambda$, as in ridge regression.
The error covariance matrix $\E[(\hat\beta-\beta)(\hat\beta-\beta)^T]$, for a given design $X$, is $\sigma^2(\phi(X)^T\phi(X)+\pcov)^{-1}$.

Minimizing the covariance matrix of $\hat\beta$ with respect to the design $X\in\Omega^k$ corresponds to minimizing the inverse of the information matrix $\phi(X)^T\phi(X)+\pcov$, in the sense of the Loewner order. But, since the Loewner order is a partial order on the space $S_p(\R)^+$ of positive $p\times p$ symmetric matrices, the variance of $\hat\beta$ does not necessarily admit a global optimum. As a proxy, it is thus common to minimize a decreasing convex function $h:S_p(\R)^+\rightarrow\R$ of the information matrix.

In this paper, we focus on two of the most common proxies, $h_A(M)\defeq\tr(M^{-1})$ and $h_D(M)\defeq\det(M^{-1})$, respectively called the A and D-optimality criteria \citep{Pukelsheim}.
An optimal design is thus defined as an element of
\begin{equation} \label{ODproblem}
\argmin_{X\subset\Omega^k} h(\phi(X)^T\phi(X)+\pcov).
\end{equation}
Even though $h$ is a convex function, this is not the case for the set of information matrices $\{\phi(X)^T\phi(X)+\pcov, X\subset\Omega^k\}$, making the optimization problem difficult.
Even when $\Omega$ is a finite space with $n$ elements, \eqref{ODproblem} writes as a finite optimization problem over a space with $\binom{n}{k}$ elements, which is usually too large for exhaustive search.
A common technique is to solve instead a convex relaxation of \eqref{ODproblem},
\begin{equation} \label{ODproblem_relax}
\argmin_{\nu\in\mathcal{M}(\Omega)} h(G_\nu(\phi)+\pcov)~s.t.~\nu(\Omega)=k,
\end{equation}
where $\mathcal{M}(\Omega)$ is the space of Borel measures on $\Omega$, and $G_\nu(\phi)\defeq\left(\int_\Omega \phi_i(x)\phi_j(x)\der\nu(x)\right)\in\mathbb{R}^{p\times p}$
is the \emph{Gramian} matrix. A solution of the convex optimization problem \eqref{ODproblem_relax} is called an \emph{approximate optimal design}.

\subsection{The case of discrete design spaces}
\label{sec:DiscreteDS_talk}
As an illustration of the difficulty of finding an optimal design when $\Omega$ is a finite set, \cite*{D_opt_NP_Hard} show that the problem of getting a (1+$\epsilon$)-approximation of the D-optimal design is NP-hard for small enough $\epsilon$ when $\pcov=0$ (i.e., in the non-Bayesian setting). A similar result for A-optimality, when $k=p$, appears in \citep{Nikolov}.
In contrast, the relaxation \eqref{ODproblem} becomes a convex optimization problem over a finite-dimensional space, for which efficient algorithms exist \cite[Chapter 7.5.2]{Boyd}.
A natural question is thus how to extract a near-optimal design from an approximate design in \eqref{ODproblem}.
Several \emph{rounding algorithms} were introduced for that purpose \cite[Chapter 12]{Pukelsheim}, such as the popular method of \cite{Efficient_rounding}, which looks for the best design across the ones obtained by rounding up or down the coefficients of an approximate design.

Still with $\pcov=0$ in \eqref{eq:GM_estimator}, \cite{Nikolov} and \cite{Def_RescaledVS} introduced an alternative to rounding, called ``proportional" or ``rescaled" volume sampling.
The principle is to sample a random design $X$ with a probability density proportional to $\det(\phi(X)^T\phi(X))\der\nu^k(X)$, where $\nu$ is an approximate optimal design.
\cite{Nikolov} showed that for the A-optimality criterion, the generated designs are on average at least a $k/(k-p+1)$-approximation of the A-optimal design.
\cite{RegDPP} further introduced the \emph{regularized determinantal point process} distribution as a way to generalize these results to the Bayesian setting, yielding $(1+O\big(\sqrt{\frac{\ln(k)}{k}}\big))$-approximations of the optimal design for several optimality criteria, including A and D, provided $k$ is large enough.

\subsection{The case of continuous design spaces} \label{sec:ContinuousDS_talk}


Theory is well-developed in the non-Bayesian univariate case, when $\Omega$ is a compact subset of $\R$ and $\pcov=0$. The approximate D-optimal design has been characterized as zeros of certain families of polynomials for both polynomial \cite[Theorem 5.5.3]{Dette} and trigonometric regression \cite[Theorem 3.1.]{Trigo_regression}.
Unfortunately, these results do not extend well to multivariate regression. When $\Omega$ is a compact subset of $\mathbb{R}^d$, some results do exist, but only for very specific design spaces and regression functions \citep*{Multi_Designs,Liski}. Theoretical knowledge in the Bayesian case is even more limited but a few results on D-optimality in some specific cases are known \citep*{Dette_Bayes, Bayes_Example}.

Regarding the relaxation \eqref{ODproblem_relax} of the optimal design problem, common practice is to optimize it only over finitely supported measures \citep*{Dette, DeCastro}, and to look for an optimal design within the support of a solution to that relaxation.
This is partially motivated by the result that the relaxed problem \eqref{ODproblem_relax} always has a finitely supported solution; see e.g. \cite[Theorem 8.2]{Pukelsheim}.
We refer to \citep{Pronzato} for a survey of optimization methods for the relaxation \eqref{ODproblem_relax}.
As a recent example of this line of research, \cite{DeCastro} give an algorithm solving \eqref{ODproblem_relax} when $\pcov=0$, $h$ is either the A or D-optimality criterion, $\Omega$ is a closed semi-algebraic set, and the regression functions are multivariate polynomials.
The algorithm is based on an elegant construction of a nested sequence of convex optimization problems where the search space is formed by the moments of the target measure. Unfortunately, because of the quickly exploding size of these optimization problems, along with some numerical instability issues, the practical impact of this algorithm has so far remained limited to small $p$ and $d$. In particular, experiments in \citep{DeCastro} consider polynomial regressions with degree up to $3$ in dimension $d\leq 3$.
The extent of this domain of applicability is confirmed by our own numerical experiments; see Section~\ref{sec:perf}.

In parallel, a lot of effort has been put into designing efficient optimization heuristics for optimal design, although not many algorithms are agnostic to the properties of $\Omega$.
We describe two such families in Section~\ref{sec:algo}: local search algorithms \citep[Chapter 4]{LocSearch} and the exchange method \citep[Chapter 9.2.1]{Pronzato}.

\section{Proportional volume sampling in general spaces}
\label{sec:GPCVS}
The distribution known as \emph{regularized determinantal point processes} was introduced by \cite{RegDPP} to generalize the proportional volume sampling distribution of \cite{Nikolov} to the finite Bayesian setting. In this section, we further extend the definition and results of \citep{RegDPP} to any Polish design space $\Omega$. To avoid later confusion with determinantal point processes in our paper, we stick to the original name of \emph{proportional volume sampling}, often shortened as PVS, for our distribution.

Other than treating general design spaces, notable differences of this section with \citep{RegDPP} include a more general and slightly tighter bound on the average A- and D-optimality criteria of designs sampled from PVS, conditioned on having a fixed size.
We also give an algorithm to sample from (possibly size-constrained) PVS that does not require rejection sampling, unlike \citep[Lemma 13]{RegDPP}.

\subsection{Definition}
We define proportional volume sampling (PVS) as a point process, i.e., a distribution over the space $\cup_{k\geq 0}\Omega^k$ of finite point configurations in $\Omega$.
For notions related to point processes, such as Janossy densities and correlation functions, we refer to the standard reference \citep{DVJ}.
PVS is further parameterized by a probability measure $\nu$ and $k$ independent functions $\phi_1,\dots,\phi_k$ in $L^2(\Omega,\nu)$. For

\begin{defin} \label{def:RDPP}
Let $A$ be a $p\times p$ non-negative symmetric matrix; let $\nu$ be a finite measure on $(\Omega,\mathcal{B}(\Omega))$ such that $G_{\nu}(\phi)$ is non-singular, where $\mathcal{B}(\Omega)$ is the Borel $\sigma$-algebra and $G_\nu(\phi)\in\mathbb{R}^{p\times p}$ has entries $\int_\Omega \phi_i(x)\phi_j(x)\der\nu(x)$.
We define the proportional volume sampling distribution $\PVS{\nu}{\phi}{\pcov}$ with reference measure $\nu$ as the point process on $\Omega$ with Janossy densities
\begin{equation}\label{DefPCVS}
\forall n\in\N,~\forall x\in\Omega^n,~j_n(x_1,\cdots,x_n)\der x_1\cdots\der x_n=\frac{\det(\phi(x)^T\phi(x)+\pcov)}{\det(G_\nu(\phi)+\pcov)\exp(\nu(\Omega))}\der\nu(x_1)\cdots\der\nu(x_n).
\end{equation}
\end{defin}
Some comments are in order.
First, we recall that $j_n(x_1,\cdots,x_n)\der x_1\cdots\der x_n$ can be interpreted as the probability of the point process consisting of $n$ points, one in the neighborhood of each $x_i$  \citep[Section 5.3]{DVJ}.
It follows that PVS strikes a balance between favoring designs that lie in regions of large mass under $\nu$, and designs with feature vectors $(\phi_i(x_j)) \in\mathbb{R}^n$, $i=1,\dots,k$ spanning a large volume.
When $k=p$ and $\pcov=0$, the latter is equivalent to having a small-volume confidence ellipsoid for the Gauss-Markov estimator \eqref{eq:GM_estimator}.
Second, taking $\Omega$ to be a finite set leads to a similar distribution, although not identical, to the \emph{regularized determinantal point process} defined in \citep{RegDPP}. Third, when taking $\pcov=0$ and conditioning PVS to have $k$ points, this distribution corresponds to the \emph{volume-rescaled sampling} defined in \citep{VRS}, whose discrete version itself corresponds to the \emph{proportional volume sampling with hardcore distribution} of \cite{Nikolov} and the \emph{rescaled volume sampling} described in \citep{Def_RescaledVS}.
Third, Definition~\ref{def:RDPP} bears resemblance to that of determinantal point processes \citep{Mac75,HKPV06}, an observation that we make precise in Section~\ref{sec:sampling}.
Fourth, the fact that \eqref{DefPCVS} is well-defined can be seen by showing that
$$\sum_{n\geq 0}^\infty\frac{1}{n!}\int_{\Omega^n}j_n(x)\der ^n x=1;$$
see \citep[Proposition 5.3.II.(ii)]{DVJ}. This can be done using both the classical Cauchy-Binet formula and a generalization of it \citep{Cauchy_Binet}, combined with the well-known matrix identity \citep{Det_Diag} stating that for any $n\times n$ matrix $M$ and diagonal matrix $D$ with diagonal entries $\lambda_1,\cdots,\lambda_n$,
$$\det(M+D)=\sum_{S\subset\{1,\cdots,n\}}\det(M_S)\prod_{i\notin S}\lambda_i$$
where $M_S$ is the submatrix of $M$ of the rows and columns indexed by $S$. The detailed proof is given in Section \ref{proof:RDPP}.

\subsection{PVS and optimal designs} \label{sec:PVS_results}
\cite{RegDPP} show that their regularized DPP distribution gives natural unbiased estimators of $(G_\nu(\phi)+\pcov)^{-1}$ and $\det(G_\nu(\phi)+\pcov)^{-1}$ in the discrete case. We show that this property extends to PVS defined on a general Polish spaces $\Omega$.

\begin{prop} \label{prop:Equal_Expectations}
Let $X\sim \PVS{\nu}{\phi}{\pcov}$, then
\begin{equation}\label{Eq_regDPP1}
\E\left[(\phi(X)^T\phi(X)+\pcov)^{-1}\right]=(G_\nu(\phi)+\pcov)^{-1}
\end{equation}
and
\begin{equation}\label{Eq_regDPP2}
\E\left[\det(\phi(X)^T\phi(X)+\pcov)^{-1}\right]=\det(G_\nu(\phi)+\pcov)^{-1}
\end{equation}
\end{prop}
\begin{proof}
The proof is given in Section \ref{proof:Equal_Expectations}. It relies on the fact that for any function $f:\cup_{n\geq 0}\Omega^n\rightarrow\R$,
\begin{equation}\label{eq:Expectation_PP}
\E[f(X)]=\sum_{n\geq 0}\frac{1}{n!}\int_{\Omega^n}f(x)j_n(x)\der^n x;
\end{equation}
see \citep[Exercice 5.3.8]{DVJ}.
\end{proof}

Proposition \ref{prop:Equal_Expectations} shows that by taking $\nu$ to be a solution of the relaxation \eqref{ODproblem_relax} the A or D-optimal design problem, the designs generated by PVS will perform, on average, at least as well as an A or D-optimal design. The obvious issue here is that PVS generates designs with a random number of points. When $\nu(\Omega)=k$ we later show in Corollary \ref{corr:Avg_num_points} that the average number of points generated by $\PVS{\nu}{\phi}{\pcov}$ is
$$k+\tr(G_\nu(\phi)(G_\nu(\phi)+\pcov)^{-1})\in[k,k+p].$$
When $k$ is large, this shows that the distribution $\PVS{\nu}{\phi}{\pcov}$ generates designs with $\approx k$ points.
\cite{RegDPP} make such a statement precise for their variant of PVS using concentration inequalities.
For $k$ large enough, they actually bound the average optimality criteria of designs generated by PVS, conditionally to their size being either equal to $k$, or lower than $k$ and completed with random i.i.d. points to reach size $k$.
Instead, we focus on A and D-optimality and we directly exploit the form of the Janossy densities of our PVS variant to get a tighter bound, which can be used for any value of $k\geq p$.

\begin{prop}\label{prop:regDPP_Dopt}
Let $X\sim \PVS{\nu}{\phi}{\pcov}$ such that $\nu(\Omega)=k$. Then
\begin{align}
\E \big[\det(\phi(X)^T\phi(X)+&\pcov)^{-1}\big| |X|=k\big]\nonumber\\
& \leq\frac{k^p(k-p)!}{k!}\frac{\det(G_\nu(\phi)+\pcov)^{-1}}{1+\frac{p-1}{k-p+1}\big[1-\det(G_\nu(\phi)(G_\nu(\phi)+\pcov)^{-1})\big]}
\label{Eq_regDPP2_Dbound}
\end{align}
with equality when $\pcov=0$.
\end{prop}
\begin{proof}
The proof is given in Section \ref{proof:regDPP_Dopt}. It relies on the fact that for any function $f:\Omega^k\rightarrow\R$,
$$\E[f(X)| |X|=k]=\frac{\E[f(X)\cara{|X|=k}]}{\E[\cara{|X|=k}]}=\frac{\frac{1}{k!}\int_{\Omega^k}f(x)j_k(x)\der^k x}{\frac{1}{k!}\int_{\Omega^k}j_k(x)\der^k x}.$$
This can be seen as a direct consequence of the definition of Janossy functions or  \eqref{eq:Expectation_PP}.
\end{proof}

Proposition~\ref{prop:regDPP_Dopt} shows that if we can find a solution $\nu_\star$ to the relaxation \eqref{ODproblem_relax} of the D-optimal design problem, a design sampled from $\PVS{\nu_\star}{\phi}{\pcov}$, conditionally on having $k$ points, is in expectation at least a $\frac{(k-p)!k^p}{k!}$-approximation of the (intractable) D-optimal design.
In that sense, PVS can be seen as a rounding method in the sense of \cite[Chapter 12.4]{Pukelsheim}: it takes as input a solution $\nu_\star$ to the relaxed design problem, and outputs a design in the desired form of a set of $k$ points with an optimality certificate, here in expectation.
The second denominator in the RHS of \eqref{Eq_regDPP2_Dbound} results from a second-order approximation in the proof.
We kept it to show the influence of the prior covariance matrix $\Lambda^{-1}$.
At its lowest value, this term is equal to $1$ in the non-Bayesian case ($\pcov=0$), and it grows to $(1+\frac{p-1}{k-p+1})^{-1}$ as $\pcov^{-1}$ goes towards $0$.
Our bound thus improves as the prior becomes more peaked.

To interpret the factor in Proposition~\ref{prop:regDPP_Dopt}, it is actually more convenient to measure the performance of an experimental design by its D-efficiency \citep[Chapter 11.1]{Atkinson}
\begin{equation}\label{def:D-eff}
D_{\mathrm{eff}}(X)=\left(\frac{\det(\phi(X)^T\phi(X)+\pcov)}{\det(\phi(X_{\star})^T\phi(X_{\star})+\pcov)}\right)^{1/p} \in [0,1],
\end{equation}
where $X_{\star}$ is a D-optimal design.
Since $x\mapsto x^{-1/p}$ is a convex function, it comes
\begin{multline*}
\E_{\PVS{\nu}{\phi}{\pcov}}[D_{\mathrm{eff}}(X)]\geq \left(\frac{k!}{(k-p)!k^p}\right)^{1/p}\times\\
\left(1+\frac{p-1}{k-p+1}\big(1-\det(G_\nu(\phi)(G_\nu(\phi)+\pcov)^{-1})\big)\right)^{1/p}\left(\frac{\det(G_\nu(\phi)+\pcov)}{\det(\phi(X_{\star})^T\phi(X_{\star})+\pcov)}\right)^{1/p}.
\end{multline*}
Since $\det(G_{\nu_\star}(\phi)+\pcov)\geq \det(\phi(X_{\star})^T\phi(X_{\star})+\pcov)$, designs sampled from $\PVS{\nu}{\phi}{\pcov}$ have, on average, a D-efficiency greater than $\left(\frac{k!}{(k-p)!k^p}\right)^{1/p}\geq 1-(p-1)/k$. This improves upon the results of \cite[Lemma 13]{RegDPP}, who found a D-efficiency of $1-O\big (\sqrt{\frac{\ln(k)}{k}}\big )$.

We now show that designs generated by PVS with some reference measure $\nu$ always have a better expected D-optimality criterion than i.i.d. designs drawn from $\nu/\nu(\Omega)$.

\begin{prop}\label{prop:PVS_better}
Let $\nu$ be any finite measure on $\Omega$, $X\sim\PVS{\nu}{\phi}{\pcov}$ and $Y=(Y_1,\cdots,Y_k)$ where the $Y_i$ are i.i.d. random variables with distribution $\nu/\nu(\Omega)$. Then
$$\E[\det(\phi(X)^T\phi(X)+\pcov)^{-1}\big| |X|=k]\leq \E[\det(\phi(Y)^T\phi(Y)+\pcov)^{-1}]$$
\end{prop}

We also show a bound for the A-optimality criterion of designs generated from the $\PVS{\nu}{\phi}{\pcov}$ distribution.
\begin{prop}\label{prop:regDPP_Aopt}
Let $X\sim \PVS{\nu}{\phi}{\pcov}$ such that $\nu(\Omega)=k$. We define
$$m_0\defeq\max(\dim(\mathrm{Ker}(\pcov)),1).$$
Then
\begin{equation}\label{Eq_regDPP2_Abound}
\E\left[\tr\big((\phi(X)^T\phi(X)+\pcov)^{-1}\big)\big| |X|=k\right]\leq\frac{k^{p+1-m_0}(k-p)!}{(k+1-m_0)!}\tr\big((G_\nu(\phi)+\pcov)^{-1}\big).
\end{equation}
\end{prop}
\begin{proof}
The proof is similar to Proposition \ref{prop:regDPP_Dopt}, and can be found in Section \ref{proof:regDPP_Aopt}.
\end{proof}

Proposition~\ref{prop:regDPP_Aopt} shows that if we can find a solution $\nu_\star$ of the relaxation \eqref{ODproblem_relax} of the A-optimal design problem, a design sampled from $\PVS{\nu_\star}{\phi}{\pcov}$, conditionally to having $k$ points, is in expectation a $\frac{k^{p+1-m_0}(k-p)!}{(k+1-m_0)!}$-approximation of the (intractable) D-optimal design. Defining A-efficiency as
\begin{equation}\label{def:A-eff}
A_{\mathrm{eff}}(X)\defeq\frac{\tr\big((\phi(X_\star)^T\phi(X_\star)+\pcov)^{-1}\big)}{\tr\big((\phi(X)^T\phi(X)+\pcov)^{-1}\big)}\in[0,1],
\end{equation}
where $X_{\star}$ is an A-optimal design, we get that designs sampled from $\PVS{\nu_\star}{\phi}{\pcov}$ have, on average, an A-efficiency greater than $\frac{(k+1-m_0)!}{k^{p+1-m_0}(k-p)!}$. In the special case where $\pcov=0$, the A-efficiency becomes $\frac{k-p+1}{k}$, as was already shown in \citep[Theorem 2.9]{VRS}. This bound is actually sharp since it corresponds to the worst possible gap between the A-optimality criterion of the approximate optimal design and that of the true optimal design \citep[Theorem C.3]{Nikolov}. Unfortunately, in Bayesian linear regression with $\pcov$ invertible, the A-efficiency is $\frac{k!}{k^p(k-p)!}$, which is much larger. Still, as $k\rightarrow\infty$, the A-efficiency is equivalent to $1-p(p-1)/2k$. This improves upon the $O\big(\sqrt{\frac{\ln(k)}{k}}\big)$ rate in \citep{RegDPP}, although for small values of $k$ the bound in \citep[Lemma 13]{RegDPP} is tighter.

\subsection{Efficient simulation of proportional volume sampling}
\label{sec:sampling}

When $\Omega$ is finite, an algorithm has been proposed in \citep{RegDPP} with a sample time of $O(np^2)$, where $n$ is the cardinality of $\Omega$, after a preprocessing cost of $O(nd^2)$.
This was done by showing that sampling from PVS boils down to sampling a determinantal point process (DPP) and a few additional i.i.d. points from $\nu$.
A similar result was shown for general design spaces $\Omega$ when $\pcov=0$ in \citep{VRS}.
These results further highlight the links between these two families of point processes.
DPPs are a large class of point processes formalized by \cite{Mac75} as a fermionic analogue to photon detection in quantum optics. Since then, DPPs have been extensively studied in the literature, from random matrix theory to spatial statistics and machine learning; see e.g. \citep{Hough, Lavancier, KuTa12}.

In this section, we extend the results of \cite{RegDPP} to a general Polish space.
Additionally, for any $k\geq p$, we show that there is a natural rejection-free sampler for our PVS conditionally to the cardinality being $k$.
First, we recall the definition of a determinantal point process on a general Polish set $\Omega$.

\begin{defin}
Let $\Omega$ be a Polish set and $\nu$ be a measure on $(\Omega,\mathcal{B}(\Omega))$. Let $\psi_1,\cdots,\psi_k\in L^2(\Omega,\nu)$ be $k$ orthonormal functions, and let $\lambda_1,\cdots,\lambda_k\in[0,1]$. We further define the function $K(x,y)\defeq\sum_{i=1}^k \lambda_i\psi_i(x)\psi_i(y)$. Then the point process with correlation functions
\begin{equation} \label{def:DPP}
\rho_n(x_1,\cdots,x_n)\der x_1\cdots\der x_n=\det\left(
\begin{array}{ccc} K(x_1,x_1) & \cdots & K(x_1,x_n)\\
\vdots & \ddots & \vdots\\
K(x_n,x_1) & \cdots & K(x_n,x_n) \end{array}
\right)\der\nu^k(x), \quad \forall n\geq 1,
\end{equation}
 is well-defined. We call it the DPP with reference measure $\nu$ and kernel $K$, and denote it by $\mathrm{DPP}(K,\nu)$.
\end{defin}
Note that DPPs can be sampled by an algorithm due to \cite[Algorithm 18]{Hough}; see also \citep[Algorithm 1]{Lavancier}. Additionally, DPPs are defined through their correlation functions. In order to compare them to PVS, we first get an explicit expression of the correlation functions of proportional volume sampling.

\begin{prop}\label{prop:Corr_functions}
The $n$-th order correlation function $\rho_n$ of the $\PVS{\nu}{\phi}{\pcov}$ point process are well-defined for all $n\in\N$, and write
$$\rho_n(x_1,\cdots,x_n)\der x_1\cdots\der x_n=\frac{\det(G_\nu(\phi)+\pcov+\phi(x)^T\phi(x))}{\det(G_\nu(\phi)+\pcov)}\prod_{i=1}^n\der\nu(x_i).$$
\end{prop}
\begin{proof}
The proof is shown in Section \ref{proof:Corr_functions}. It is mainly based around the usual identity linking correlation functions and Janossy densities; see \cite[Lemma 5.4.III]{DVJ}:
$$\rho_n(x_1,\cdots,x_n)\der^n x=\sum_{m\geq 0}\frac{1}{m!}\int_{\Omega^m}j_{n+m}(x,y)\der^m y.$$
\end{proof}

As a consequence, we get that PVS can be expressed as the superposition of a DPP and a Poisson point process.
This is to be put in parallel to \cite[Theorem 2.4]{VRS} which shows a similar result in the non-Bayesian case.

\begin{prop}\label{prop:Superpos}
Let $\nu$ be a finite measure on $\Omega$. We write the spectral decomposition of the matrix $G_\nu(\phi)^{1/2}(G_\nu(\phi)+\pcov)^{-1}G_\nu(\phi)^{1/2}$ as $P^TDP$, where $D=\diag(\lambda_1,\cdots,\lambda_p)$.
Furthermore, we define the functions
$$(\psi_1(x),\cdots,\psi_p(x))\defeq (\phi_1(x),\cdots,\phi_p(x))G_{\nu}(\phi)^{-1/2}P.$$

Let $X$ be a DPP with kernel
\begin{equation*}\label{Kernel_superpos}
K(x,y)\defeq\sum_{i=1}^p \lambda_i\psi_i(x)\psi_i(y)
\end{equation*}
and reference measure $\nu$, and let $Y$ be an independent Poisson point process with intensity $\nu$. Then the distribution of the superposition $X\cup Y$ is $\PVS{\nu}{\phi}{\pcov}$.
\end{prop}
\begin{proof}
The proof is given in Section \ref{proof:Superpos}. We rewrite the correlation functions of PVS as $\det(I_n+(K(x_i,x_j))_{1\leq i,j\leq n})$, which corresponds to the correlation functions of the superposition of a DPP and a Poisson point process.
\end{proof}
Note that when $\nu$ is not a diffuse measure, the Poisson point process with intensity $\nu$ can generate the same point multiple times. In this case, $X$ and $Y$ are understood as multisets and their superposition is understood as a multiset union. We also remark that by definition of the functions $\psi_i$,
$$G_\nu(\psi)=P^TG_{\nu}(\phi)^{-1/2}G_{\nu}(\phi)G_{\nu}(\phi)^{-1/2}P=I_p,$$
so that $\psi_1,\cdots,\psi_p$ is an orthonormal family of $L^2(\Omega,\nu)$. Moreover, since
$$0_p\leq G_\nu(\phi)^{1/2}(G_\nu(\phi)+\pcov)^{-1}G_\nu(\phi)^{1/2}\leq G_\nu(\phi)^{1/2}G_\nu(\phi)^{-1}G_\nu(\phi)^{1/2}=I_p$$
in Loewner's order, the eigenvalues $\lambda_1,\cdots,\lambda_p$ are all in $[0,1]$, making the DPP in Proposition \ref{prop:Superpos} is well-defined. A direct consequence of this result, proved in Section \ref{proof:Avg_num_points}, is an explicit expression of the average number of points generated by PVS.

\begin{corr}\label{corr:Avg_num_points}
Let $X\sim \PVS{\nu}{\phi}{\pcov}$, then
$$\E[|X|]=\nu(\Omega)+\tr((G_\nu(\phi)+\pcov)^{-1}G_\nu(\phi)).$$
\end{corr}

Proposition \ref{prop:Superpos} yields a natural sampling algorithm for proportional volume sampling as a superposition of an independent DPP and a Poisson point process. For ease of reference, the pseudocode of the algorithm in Proposition~\ref{prop:Superpos} is given in Figure~\ref{f:algo}.

In practice, it is customary to fix the cardinality $k$ of an experimental design in advance.
We thus would like to condition the two point processes in the superposition of Proposition~\ref{prop:Superpos} on their union having total size $k$.
This is possible using a well-known decomposition of the DPP $X$ with kernel $K(x,y)\defeq\sum_{i=1}^p \lambda_i\psi_i(x)\psi_i(y)$ as a mixture of projection DPPs \citep{Hough}.
More precisely, to sample $X$, one first samples $p$ independent Bernoulli random variables $I_i\sim \text{Ber}(\lambda_i)$. Then, conditionally on these Bernoullis, one uses the chain rule to sample from the projection DPP with kernel $\sum_{i=1}^p I_i\psi_i(x)\psi_i(y)$, which yields $\sum_{i=1}^p I_i$ points almost surely.
As a consequence, conditioning PVS to having size $k$ can be performed by first sampling independent $I_i\sim\text{Ber}(\lambda_i)$ and an independent Poisson random variable $N$ with parameter $\nu(\Omega)$, all conditioned on $N+\sum_{i=1}^p I_i=k$. Then, the union of a DPP with kernel $\sum_{i=1}^p I_i\psi_i(x)\psi_i(y)$ and $N$ i.i.d. points with distribution $\nu/\nu(\Omega)$ on $\Omega$ has the same distribution as a sample from PVS conditioned to be of size $k$.
The pseudocode of this algorithm is given in Figure~\ref{f:algo_Condition_Cardinal}.
The costly steps are the one-time computation of the Gramian and its square root, and the DPP sampler.
For the Gramian, when exact computation is not possible, one has to resort to numerical integration, possibly even Monte Carlo methods \citep{RoCa04}, depending on how complicated the set $\Omega$ is, the dimension of the design space, and the regularity of the functions $\phi_i$ and the reference measure $\nu$.
Getting an approximate Gramian makes the overall algorithm heuristic.
For the DPP sampler, the number of operations is at least cubic in the number $p$ of points in the DPP sample. The \emph{at least} corresponds to the number of rejections in the $p$ rejection samplers involved in the chain rule for projection DPPs; see \citep{GaBaVa19a} for empirical investigations on these rejection numbers.
Overall, while the cost of the algorithm in Figure~\ref{f:algo} and~\ref{f:algo_Condition_Cardinal} is at least cubic in $p$, this is not a obstacle in practice. Indeed, experimental design is used in situations where obtaining the regression labels takes time or money, so that spending a few minutes finding a good design is considered negligible.

\begin{figure}
  \begin{algorithm}[H]
\SetAlgoLined
\DontPrintSemicolon
1: Compute the Gramian matrix $G_\nu(\phi)$.\;
2: Compute the spectral decomposition $G_\nu(\phi)^{1/2}(G_\nu(\phi)+\pcov)^{-1}G_\nu(\phi)^{1/2}=P^T D P$ where $D=\diag(\lambda_1,\cdots,\lambda_p)$.\;
3: Compute the orthonormalized functions $(\psi_1,\cdots,\psi_p)\defeq (\phi_1,\cdots,\phi_p)G_{\nu}(\phi)^{-1/2}P$.\;
4: Sample $X$ from a DPP$(K,\nu)$ distribution, where $K(x,y)=\sum_i \lambda_i\psi_i(x)\psi_i(y)$.\;
5: Sample $Y$ from a Poisson point process with intensity $\nu$.\;
6: Return $X\cup Y$.
\caption{A sampler for proportional volume sampling; see Proposition~\ref{prop:Superpos}.}
\end{algorithm}
\caption{An algorithm to sample PVS, using a DPP sampler as a subroutine.}
\label{f:algo}
\end{figure}

\begin{figure}
  \begin{algorithm}[H]
\SetAlgoLined
\DontPrintSemicolon
1: Compute the Gramian matrix $G_\nu(\phi)$.\;
2: Compute the spectral decomposition $G_\nu(\phi)^{1/2}(G_\nu(\phi)+\pcov)^{-1}G_\nu(\phi)^{1/2}=P^TD P$ where $D=\diag(\lambda_1,\cdots,\lambda_p)$.\;
3: Compute the orthonormalized functions $(\psi_1,\cdots,\psi_p)\defeq (\phi_1,\cdots,\phi_p)G_{\nu}(\phi)^{-1/2}P$.\;
4: Sample $p$ independent Bernoullis $I_i\sim\text{Ber}(\lambda_i)$ and an independent Poisson random variable $N$ with parameter $\nu(\Omega)$, conditionally to $N+\sum_{i=1}^p I_i = k$.\;
5: Sample $X$ from a DPP$(K,\nu)$ distribution, where $K(x,y)=\sum_i I_i\psi_i(x)\psi_i(y)$.\;
6: Sample $N$ points $Y=\{Y_1,\cdots,Y_{N}\}\subset\Omega$ i.i.d. with distribution $\nu/\nu(\Omega)$.\;
7: Return $X\cup Y$.
\caption{A sampler for proportional volume sampling conditioned on having a given cardinality $k$; see Proposition~\ref{prop:Superpos}.}
\end{algorithm}
\caption{An algorithm to sample PVS conditioned on having a given cardinality $k$, using a projection DPP sampler as a subroutine.}
\label{f:algo_Condition_Cardinal}
\end{figure}

\subsection{A numerical illustration of PVS for experimental design}
\label{s:simple_experiment}
We illustrate the theoretical results from Section \ref{sec:PVS_results} on a simple example.
We consider the design space $\Omega=[0,1]^2$ and $k=p=10$, with the regression functions $\phi_i$ being the $p$ bivariate polynomials of degree $\leq 3$, renormalized so that $\|\phi_i\|_{L^2(\Omega)}=1$ for all $i$.
For our example, we considered the case where the inverse prior covariance matrix $\pcov$ is either $I_{10}$, $10^{-2} I_{10}$ or $10^{-4} I_{10}$.
For each case, we show in Figure \ref{fig:CompDPP-U} the distribution of the D-efficiency \eqref{def:D-eff} and A-efficiency \eqref{def:A-eff} of $2000$ random designs.
The A and D-optimal designs are not known, so we use as a baseline the best design found by a hundred runs of Federov's exchange heuristic \citep{Fedorov}, which we observed to perform quite well in this example; see Section \ref{sec: Example2}.

The first and third columns in each subfigure of Figure \ref{fig:CompDPP-U} respectively correspond to designs generated uniformly with i.i.d. points, and designs sampled from $\PVS{\nu}{\phi}{\pcov}$, where $\nu$ is the uniform distribution over $\Omega$.
$\PVS{\nu}{\phi}{\pcov}$ improves over uniform i.i.d. sampling, and the improvement is larger when $\pcov$ is small, i.e., when the prior in the Bayesian linear regression has a large variance.
This can be seen as a consequence of Proposition \ref{prop:Superpos} and Corollary \ref{corr:Avg_num_points}, stating that $\PVS{\nu}{\phi}{\pcov}$ consists in an average of
$$\tr(G_\nu(\phi)(G_\nu(\phi)+\pcov)^{-1})$$
points generated from a DPP completed with i.i.d. samples. When $\pcov$ takes large values, $\PVS{\nu}{\phi}{\pcov}$ is thus very close to an i.i.d. distribution, explaining the similarity between both distributions in Figure \ref{fig:CompDPP-U_Dopt_1}, while in Figure \ref{fig:CompDPP-U_Dopt_10000}  $\PVS{\nu}{\phi}{\pcov}$ performs significantly better than its i.i.d. counterpart.
In fact, in Figure \ref{fig:CompDPP-U_Dopt_10000}, the average PVS design is on par with the best design out of thousands of uniform i.i.d. designs.
\begin{figure}
\centering
\subfloat[D-efficiency, $\pcov=I_{10}$\label{fig:CompDPP-U_Dopt_1}]{\includegraphics[width=\Vplot]{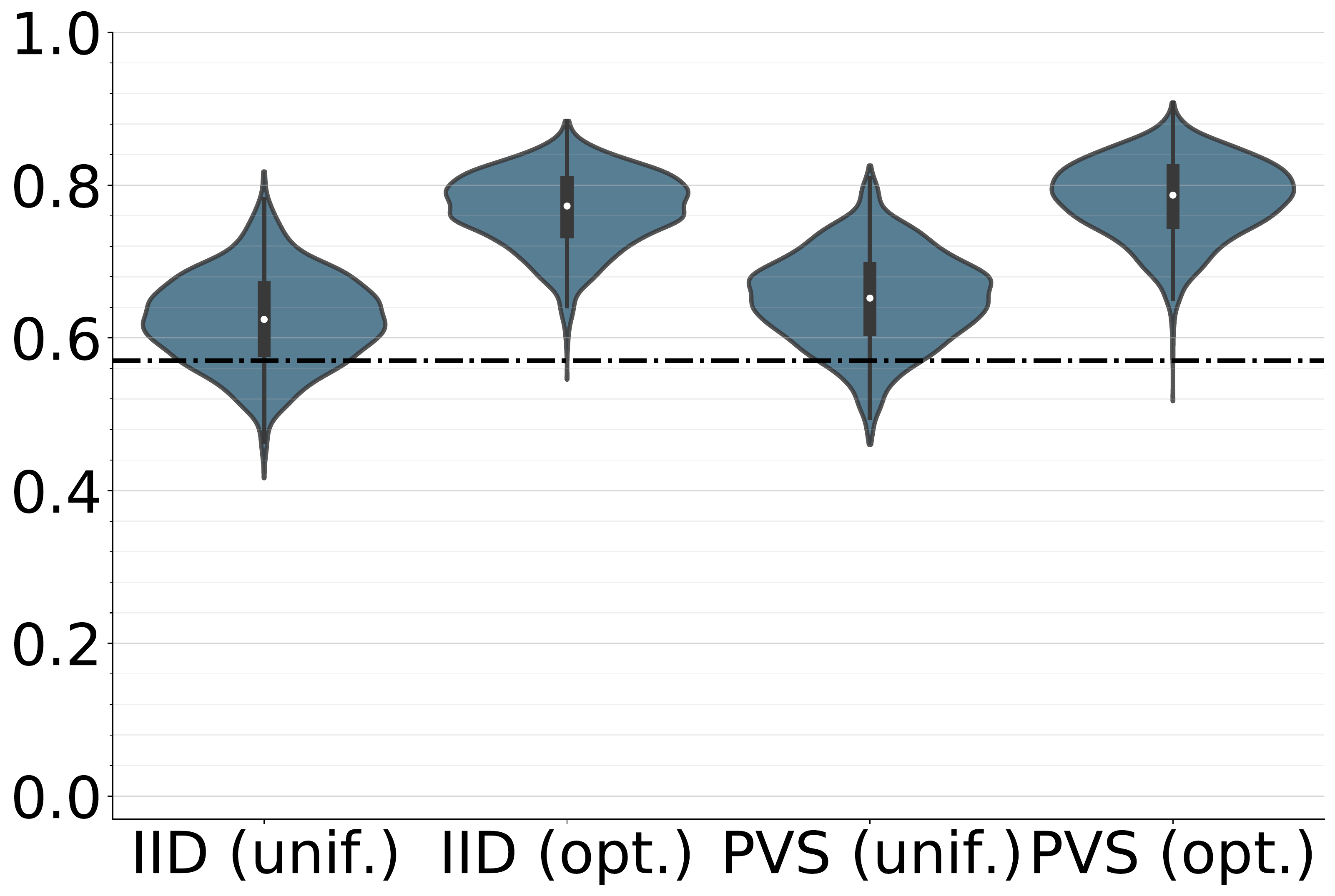}}\hfill
\subfloat[D-efficiency, $\pcov=10^{-2} I_{10}$\label{fig:CompDPP-U_Dopt_100}]{\includegraphics[width=\Vplot]{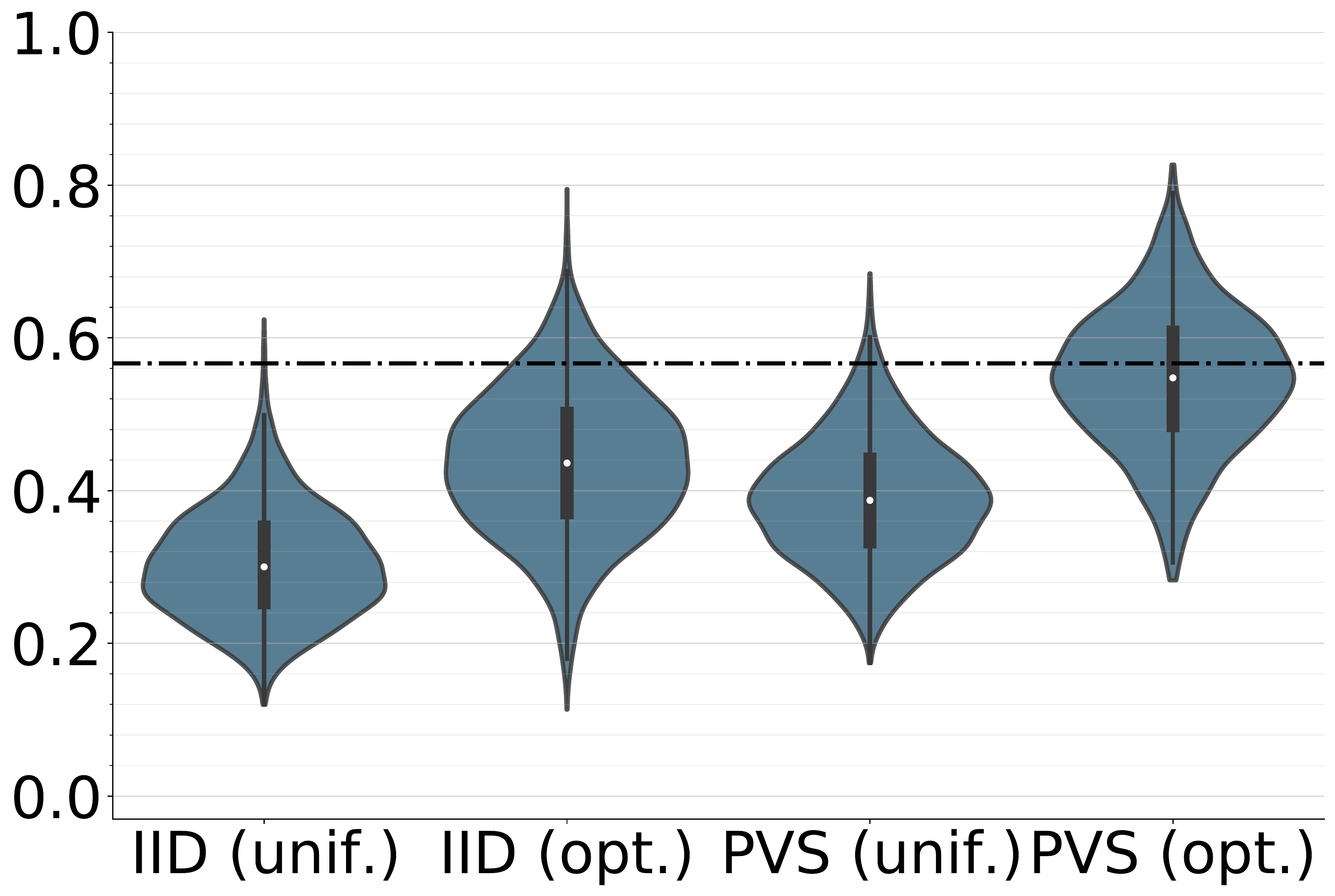}}\hfill
\subfloat[D-efficiency, $\pcov=10^{-4} I_{10}$\label{fig:CompDPP-U_Dopt_10000}]{\includegraphics[width=\Vplot]{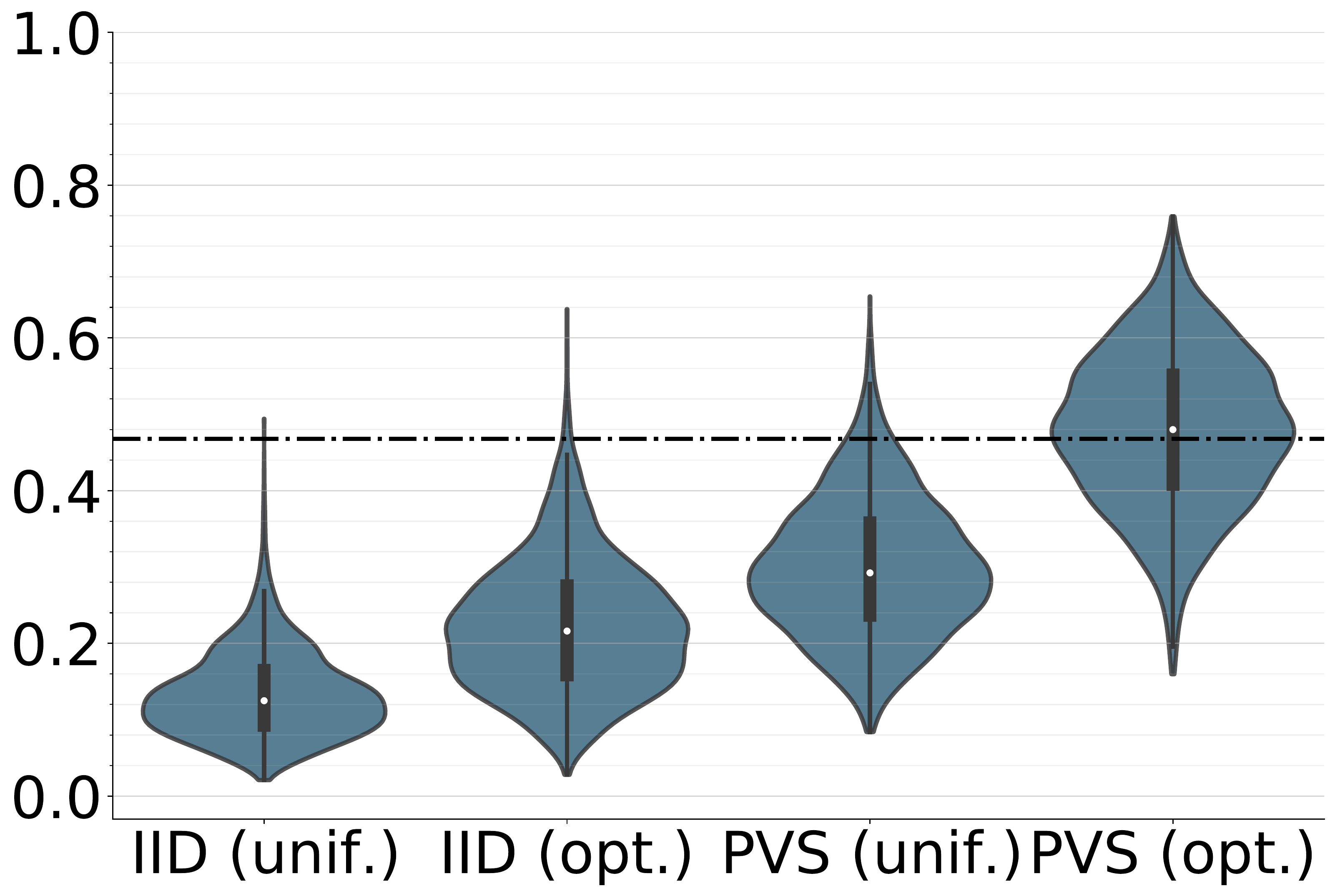}}
\hfill
\subfloat[A-efficiency, $\pcov=I_{10}$\label{fig:CompDPP-U_Aopt_1}]{\includegraphics[width=\Vplot]{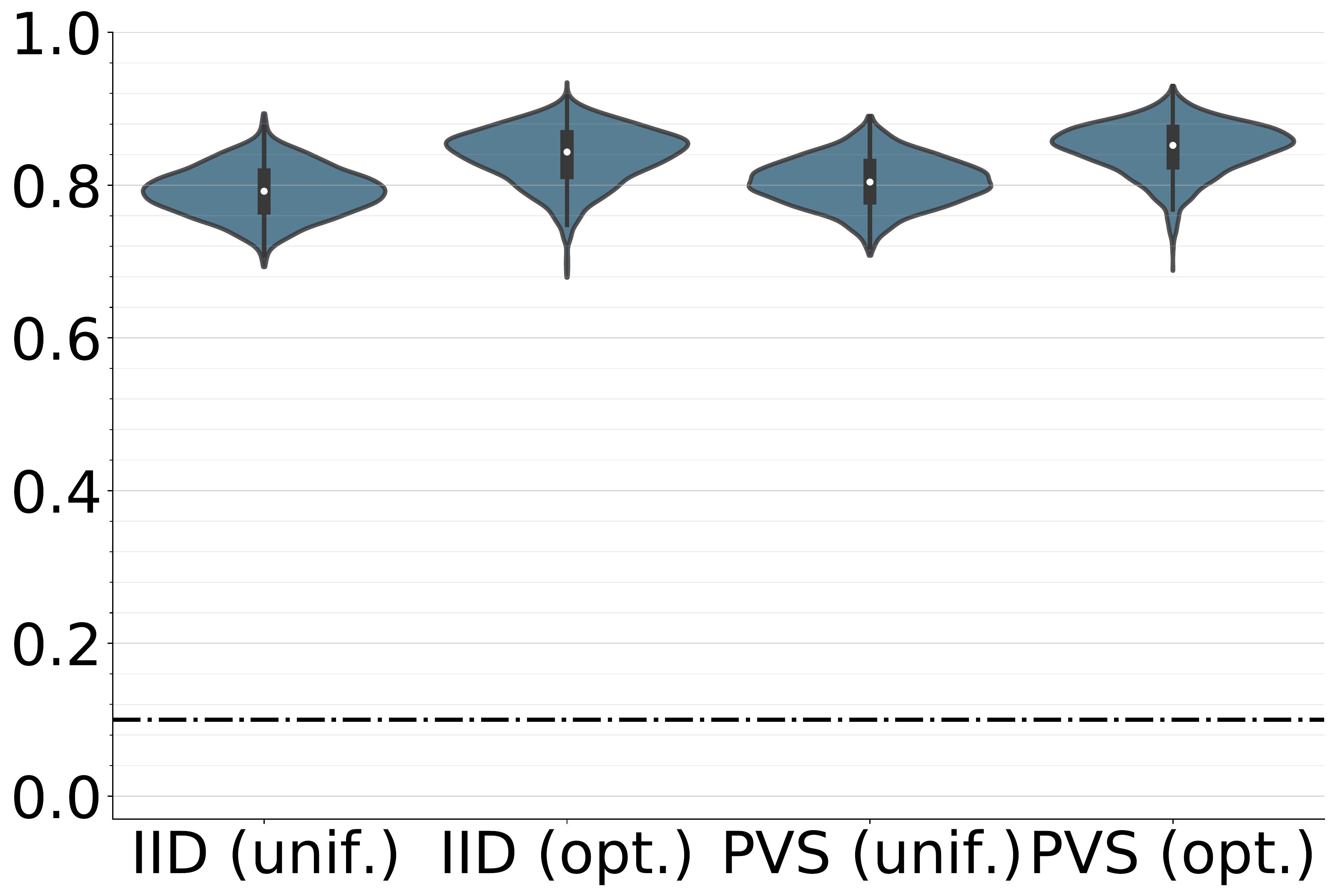}}
\hfill
\subfloat[A-efficiency, $\pcov=10^{-2} I_{10}$\label{fig:CompDPP-U_Aopt_100}]{\includegraphics[width=\Vplot]{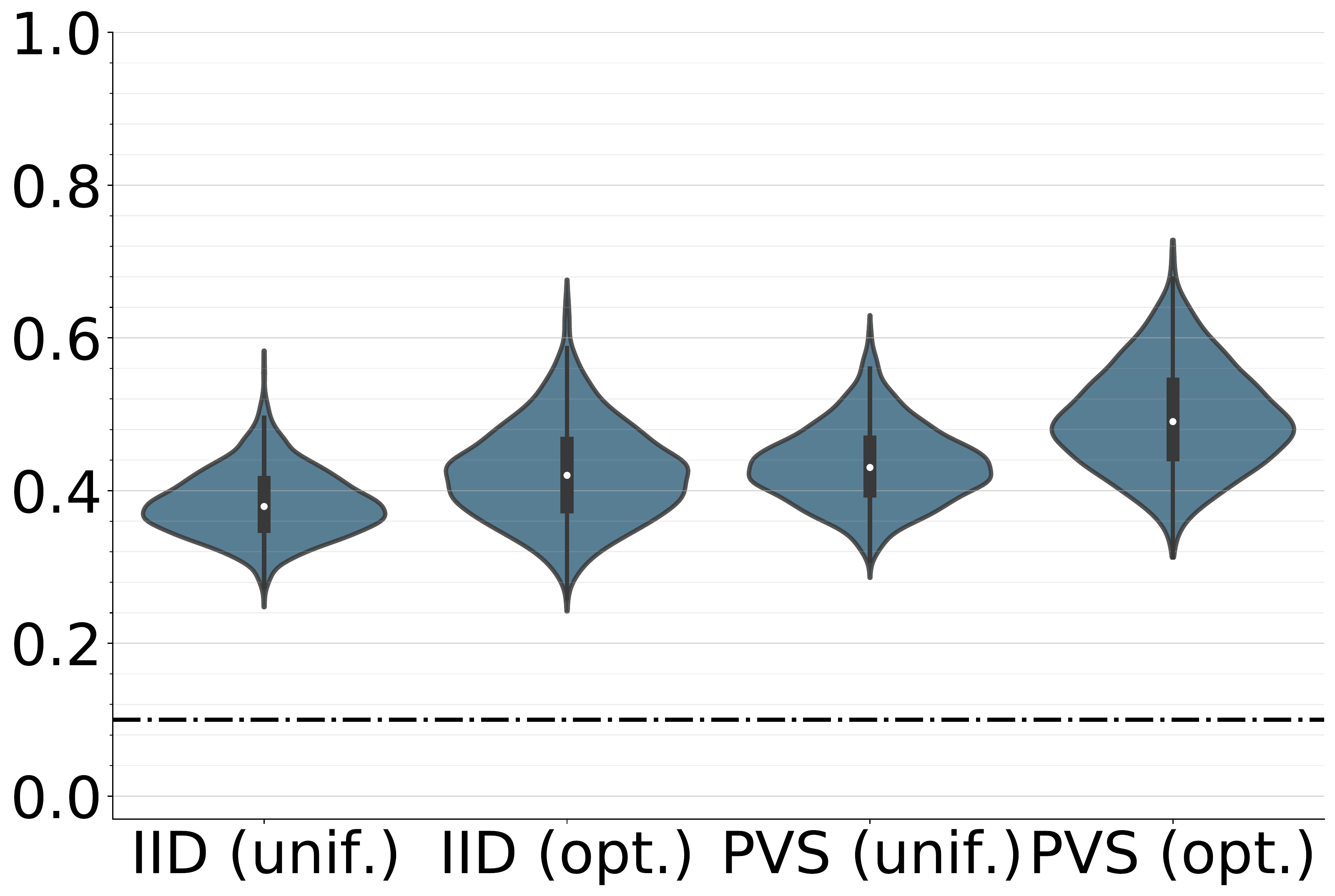}}
\hfill
\subfloat[A-efficiency, $\pcov=10^{-4} I_{10}$\label{fig:CompDPP-U_Aopt_10000}]{\includegraphics[width=\Vplot]{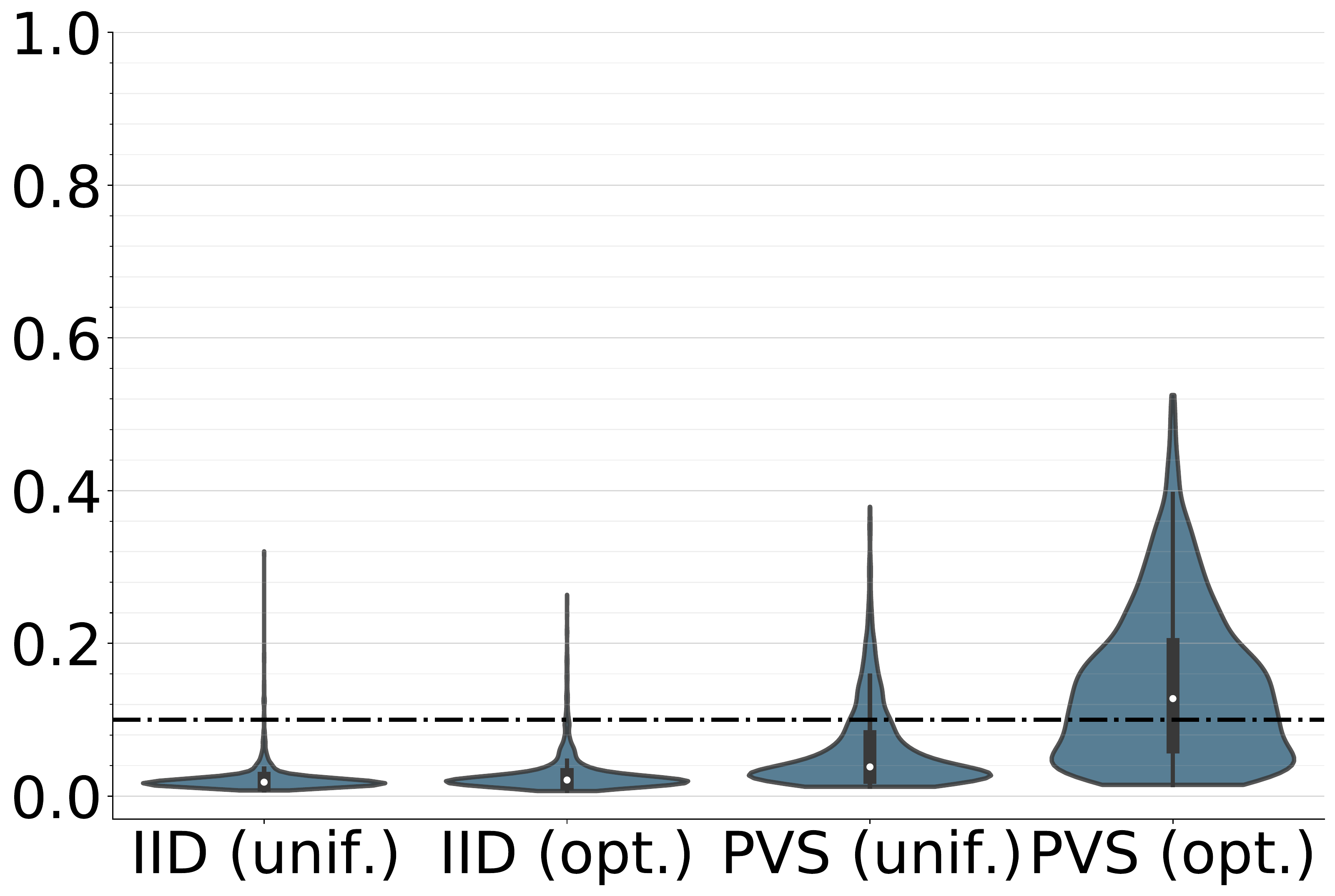}}
     \caption{{\label{fig:CompDPP-U}\small Violin and boxplots of the Bayesian A-efficiency and D-efficiency of $2000$ random designs either sampled i.i.d. from $\nu$ or from $\PVS{\nu}{\phi}{\pcov}$. The density $\nu$ is either uniform on $\Omega$ or the optimized density shown in Figure \ref{fig:Opti_Poly_Dens}.
     Dashdotted lines show the bounds of Proposition \ref{prop:regDPP_Dopt} (top) and the non-Bayesian bound of Proposition \ref{prop:regDPP_Aopt} (bottom); See Section \ref{s:simple_experiment} for details.}}
\end{figure}
Now, as discussed in Section~\ref{sec:PVS_results}, a natural way to further improve on these results is to set $\nu$ to be a solution $\nu_\star$ of the relaxed optimization problem \eqref{ODproblem_relax}.
In the non-Bayesian setting, a finitely supported approximate optimal design is already known \citep{Multi_Designs}, but this is not the case when $\Lambda\neq 0$.
In order to illustrate the results of Propositions \ref{prop:regDPP_Dopt} and \ref{prop:regDPP_Aopt}, we thus consider a simple parametrized form for the density of $\nu$, chosen to make the relaxed optimization problem convex and thus amenable to numerical solvers.
The details of our parametrization can be found in Appendix~\ref{experiment:approximation_OD}.

The second and fourth columns in each panel of Figure~\ref{fig:CompDPP-U} show the results of i.i.d. designs drawn from $\nu_\star$ and designs drawn from $\PVS{\nu_\star}{\phi}{\pcov}$, respectively.
As expected, using $\nu^\star$ further improves over PVS with uniform reference measure.
We observe that the dashdotted bound from Proposition~\ref{prop:regDPP_Dopt} is pessimistic when $\pcov$ takes large values. This is even more visible for the bound from Proposition~\ref{prop:regDPP_Aopt}, which is $10^{10}/10!\approx 0.00036$, and thus not shown in the second row of panels.
Instead, we show the bound for $\Lambda=0$ in all three lower panels: the bounds equals $0.1$, it represents well what happens at small $\Lambda$, but is still far lower than the observed average A-efficiency of designs generated from $\PVS{\nu_\star}{\phi}{\pcov}$ for large $\Lambda$.
This suggests that the bound of Proposition~\ref{prop:regDPP_Aopt} could be improved in the Bayesian case.
It is also interesting to see here the empirical confirmation that an optimized reference measure always improves over i.i.d. sampling.

In short, when looking at Figure \ref{fig:CompDPP-U} from left to right, we move from an informative prior to a frequentist setting. Simultaneously, we move from loose bounds and parity between PVS and iid designs to tighter bound for D-efficiency and big differences in performance between PVS and iid.

Finally, we show in Figure~\ref{fig:PtSimu} a few sample designs from the various distributions considered in Figure \ref{fig:CompDPP-U_Dopt_10000}. PVS naturally forces points apart from each other, showing a so-called repulsive behavior compared to i.i.d. designs.
This well-spreadedness is also characteristic of the support of optimal designs. We also observe that the density $\nu_\star$ for D-optimality, shown in Figure \ref{fig:Opti_Poly_Dens_D_10000}, takes a large value on the vertices of the square design space, a medium value on the edges and a low value in the middle.
Because of the inherent repulsiveness of the PVS distribution, this means that PVS designs will usually have a point close to each vertex, a few well-spread points close to the edges, and a few points in the center of the design space, far apart from the rest.
This behavior is very similar to the one of the D-optimal design shown in Figure \ref{fig:simu_OD}.

\begin{figure}
\centering
\subfloat[Uniform design \label{fig:simu_uni}]{\includegraphics[width=4cm]{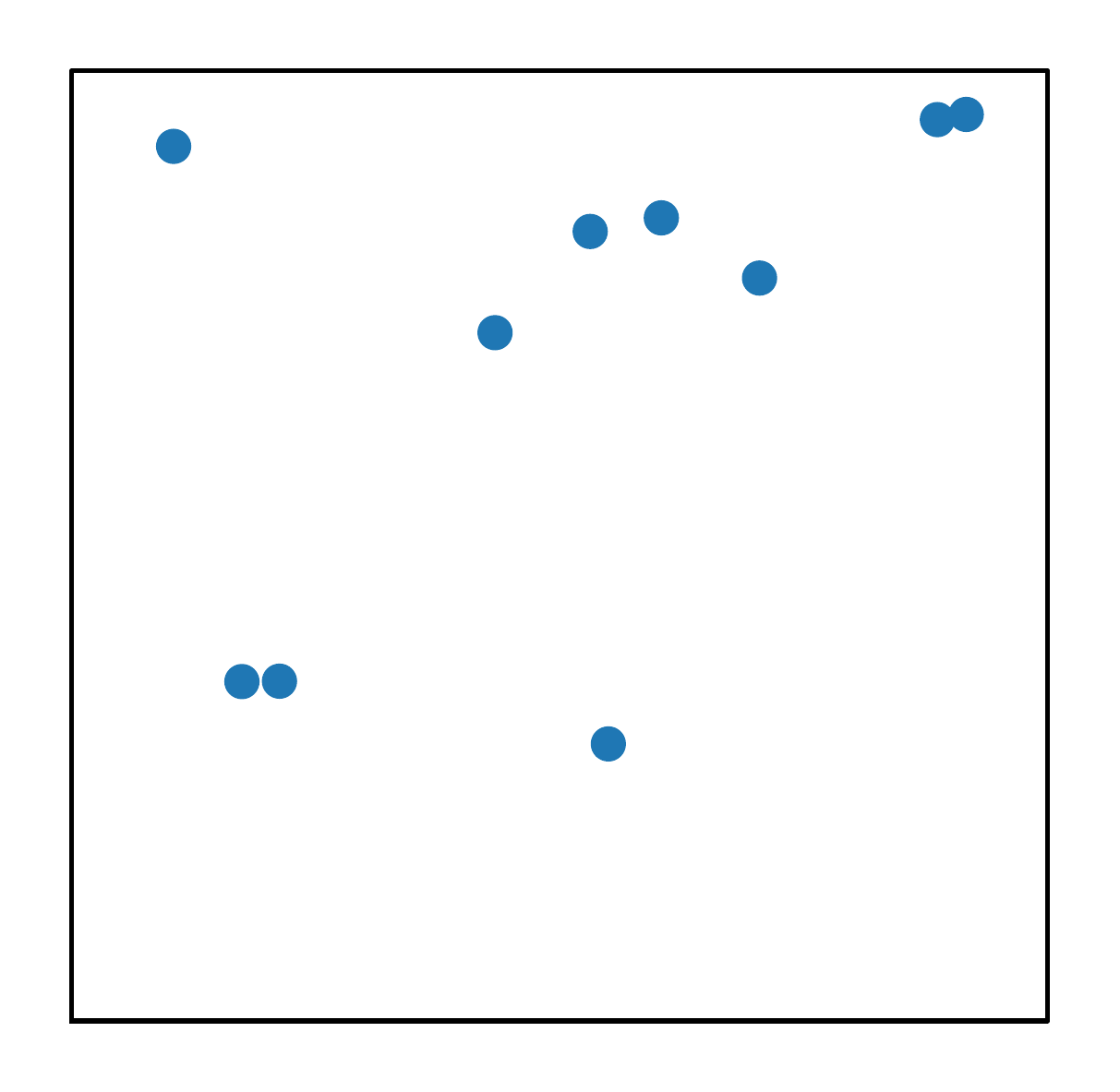}}
\hfill
\subfloat[PVS (unif.) \label{fig:simu_PVS_uni}]{\includegraphics[width=4cm]{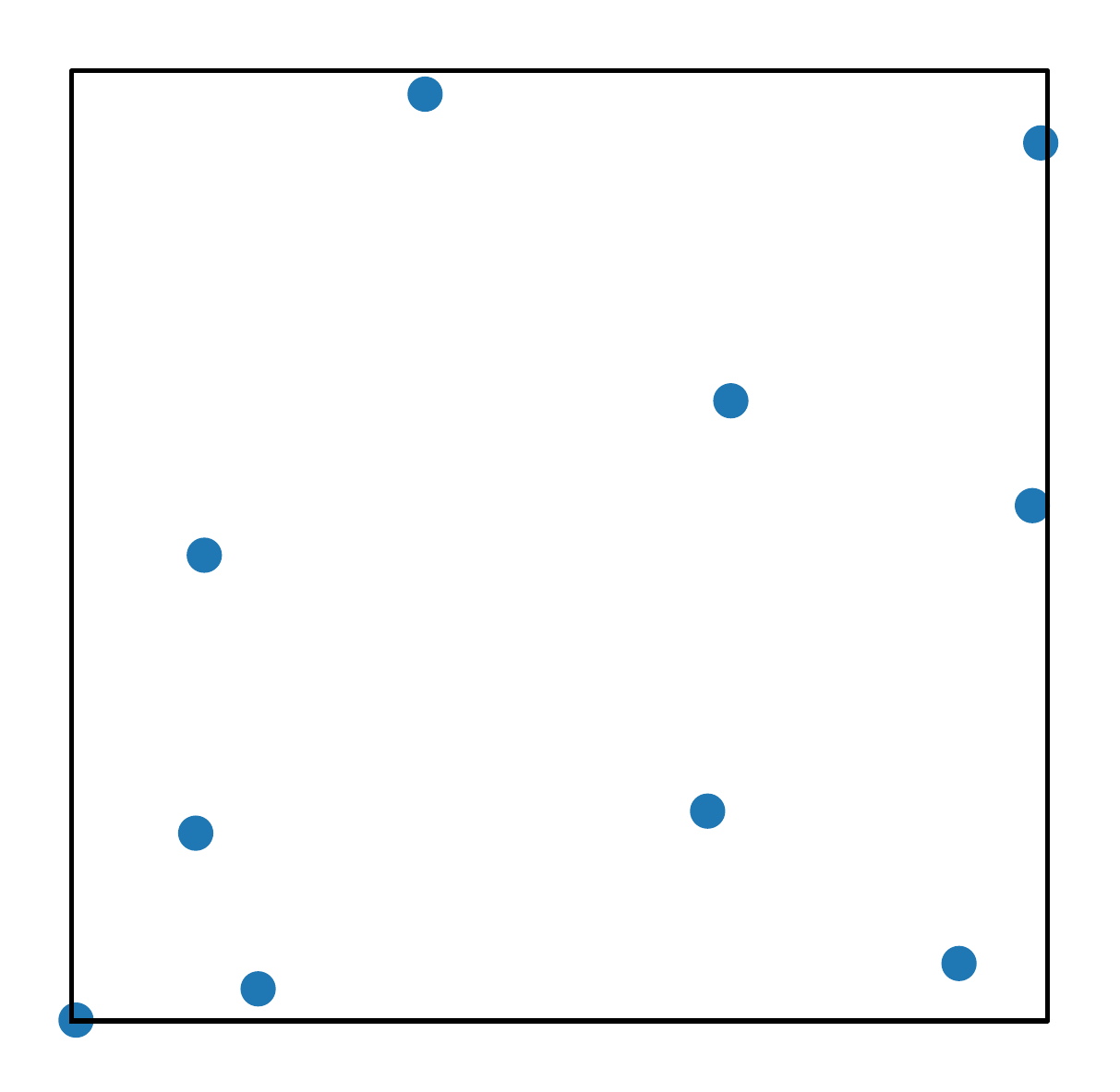}}
\hfill
\subfloat[PVS (opt.) \label{fig:simu_PVS_opt}]{\includegraphics[width=4cm]{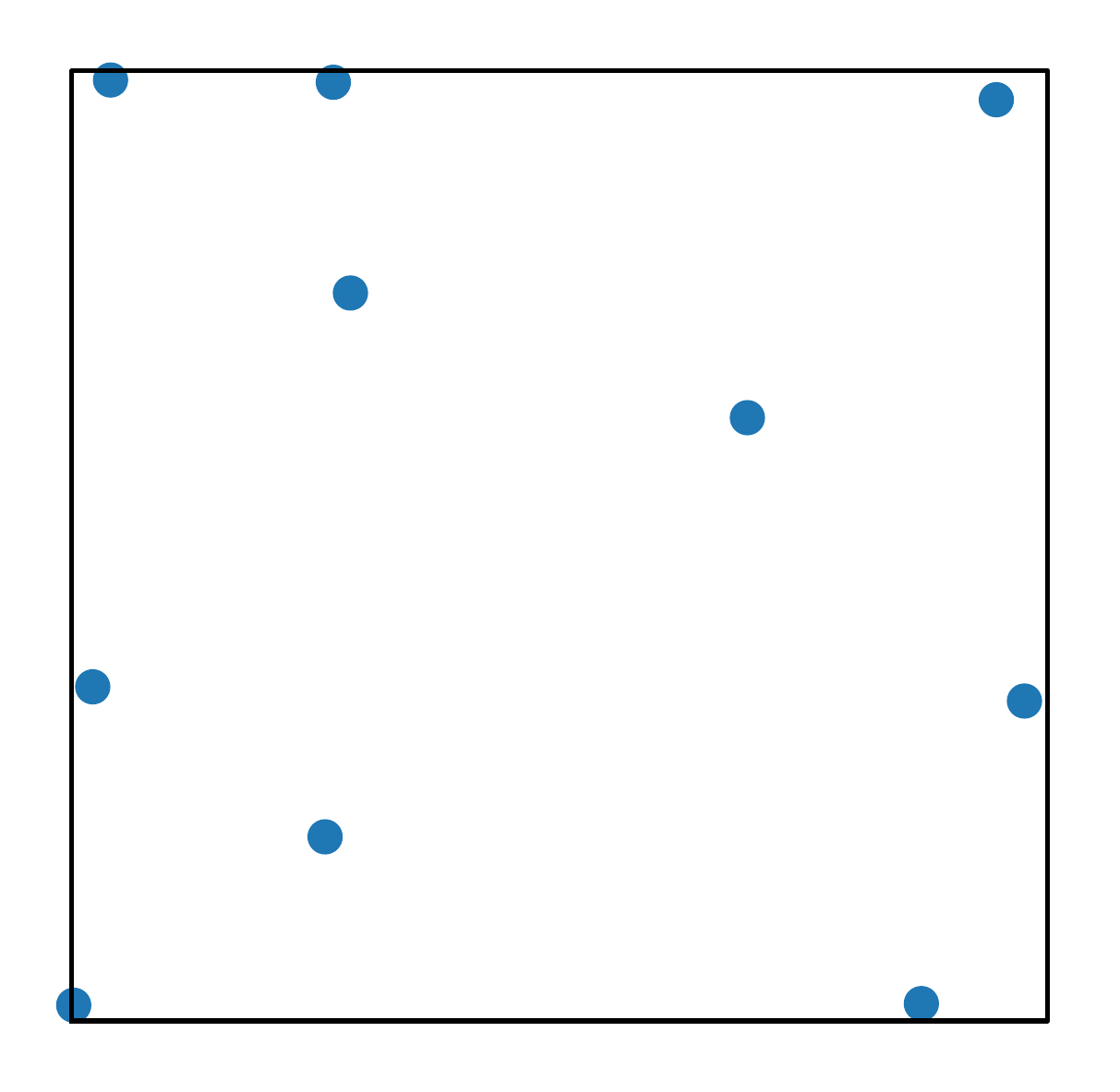}}
\hfill
\subfloat[D-Optimal design \label{fig:simu_OD}]{\includegraphics[width=4cm]{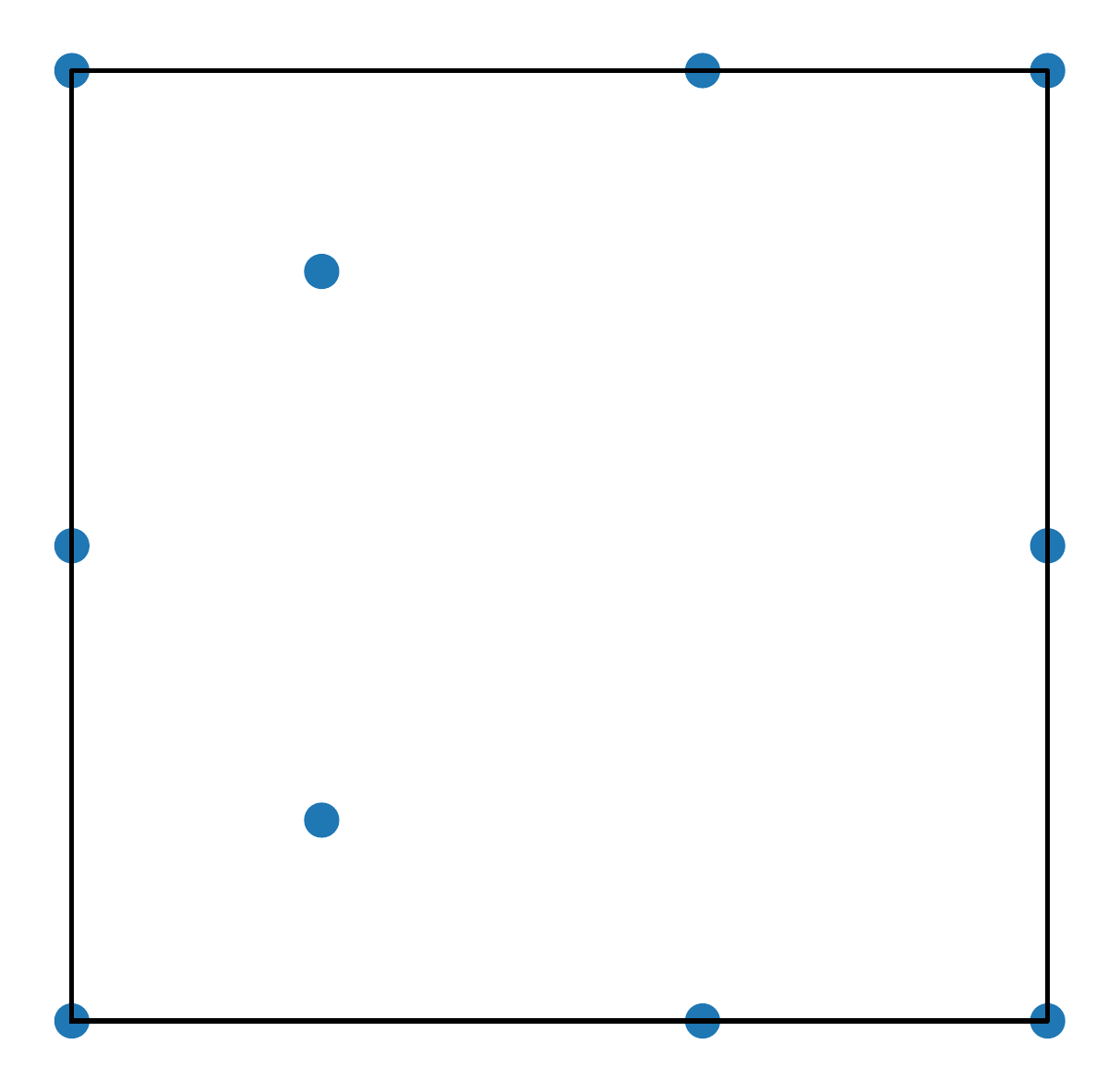}}
\hfill
	   \caption{{\label{fig:PtSimu} \small Example of simulation of a random design from either a uniform distribution (a) or the $\PVS{\nu}{\phi}{\pcov}$ distribution where $\nu$ is either the uniform distribution on $\Omega$ (b) or the distribution on $\Omega$ with the optimized density shown in Figure \ref{fig:Opti_Poly_Dens_D_10000} (c). Figure (d) shows an optimal design for comparison.}}
\end{figure}

%
%

Overall, PVS is mathematically elegant and tractable, and comes with guarantees in the form of bounds on the expected A- and D-criteria.
It also has the advantage of working for nearly any design space, any family of regression functions and for both the Bayesian and non-Bayesian setting.
A simple empirical investigation confirms that it generates designs with typically lower A- and D-optimality criteria than i.i.d. sampling, and optimizing the reference measure further helps.
On the negative side, the theoretical bounds can be loose. Furthermore, as seen in Figure~\ref{fig:CompDPP-U}, even taking the best of thousands of PVS samples does not yield a design arbitrarily close to an optimal design.
From a practical point of view, this is slightly disappointing.
This is why we now investigate a different way of using PVS, namely embedding it into an iterative search algorithm for optimal designs.

\section{Turning discrete PVS into a global search heuristic}
\label{sec:algo}
In this section, we first review two standard algorithmic templates to find approximate A and D-optimal designs in the most general setting.
Then, we introduce a new heuristic, named DOGS, that relies on the discrete proportional volume sampling of \cite{Nikolov} and \cite{Def_RescaledVS} to make global moves across a generic design space $\Omega$. Finally, we investigate the numerical performance of DOGS against standard methods.

\subsection{Standard heuristics for optimal designs}\label{sec:algo_presentation}
First, the \emph{local search algorithm} (LSA) starts with a random design drawn from some initial distribution as its current best $\Xbest$. Then, at each iteration and until a stopping criterion is met, LSA generates another random design $Y$ from a proposal distribution supported in a neighborhood of $\Xbest$, and sets $\Xbest$ to $Y$ if $Y$ has a smaller optimality criterion.
An iteration of LSA is computationally cheap, but this is at the cost of a slow empirical  convergence and a tendency to get stuck in local minima.
We refer to \citep[Chapter 4.2]{LocSearch} for a description of LSA and variants, including simulated annealing. These variants add a few parameters that need be tuned, and in our later experiments did not qualitatively change the behavior of LSA. We thus focus on the vanilla LSA algorithm.

Another standard algorithm is the \emph{exchange method} (ExM; \citep{Fedorov}). Starting with a random initial design, ExM iteratively minimizes the optimality criterion over each point of the design, keeping the rest of the design fixed, until a stopping criterion is met.
The basic operation is thus
\begin{equation} \label{eq:Toopt_by_ExM}
\argmin_{x\in\Omega} h\big(\phi(X\cup\{x\})^T\phi(X\cup\{x\}) + \pcov\big)
\end{equation}
for any design $X\subset \mathbb{R}^d$ with cardinality $k-1$. When $\Omega$ is a (small) finite set, this can be done by exhaustive enumeration or a grid search. When $\Omega$ is continuous, however, approximating \eqref{eq:Toopt_by_ExM} can be mathematically difficult or computationally heavy, depending on the dimension and shape of $\Omega$, as well as the regularity of the regression functions $\phi_i$.
ExM is also prone to converging to local minima, and it is usually recommended to take the best result out of several restarts.  
Variants of ExM are presented in \citep[Chapter 9.2.1]{Pronzato} but, as with LSA, we focus here on the vanilla ExM as a representative of its class.

\medskip

Finally, in many applications of experimental design in continuous $\Omega$, there is some knowledge or intuition of points that are likely to be in the support of an optimal design. This is particularly common when working with design spaces with a simple shape such as spheres, simplices or cubes, and for low-dimensional polynomials as regression functions; see e.g. \citep{Atkinson}.
In these cases, one usually forces search heuristics like ExM to focus on a finite set $\mathcal{C}$ of well-chosen $\emph{candidate points}$.
One approach is to run a discrete optimization algorithm such as ExM on the finite set $\mathcal{C}$, ensuring that the support of the solution is in $\mathcal{C}$.
When no useful $\mathcal{C}$ is known and $\Omega\subset\mathbb{R}^d$, it is also common practice to run ExM on the ``agnostic" finite set formed by the intersection of $\Omega$ with a regular grid.

\subsection{Discrete optimization within global search (DOGS)}

In spite of its near-optimality guarantees, proportional volume sampling as introduced in Section~\ref{sec:PVS_results} is likely not to meet the needs of practitioners.
Indeed, as shown in Figure \ref{fig:CompDPP-U}, even taking the best design out of thousands of PVS realizations does not get us close to an optimal design, while simple heuristics like ExM can output designs with smaller optimality criterion at a comparable cost.
In this section, we propose a search heuristic based on discrete volume sampling but different from PVS. The rationale is to use volume sampling to propose global moves across $\Omega$, thus avoiding getting stuck in local minima.

The pseudocode of our DOGS heuristic is given in Figure~\ref{f:DOGS}. We start from an initial random design $\Xbest$ of $k$ points drawn from any initial distribution at hand. Then, at each iteration, we find the optimal design of $k$ points among the union of the current $\Xbest$ and a random subset $\Xnew\subset\Omega^{k'}$ drawn from a proposal distribution, with a user-defined cardinality $k'\in\N$.
This step corresponds to finding a finite optimal design of $k$ points among $k+k'$, and can be solved either exactly, by enumeration, or approximately, by a search heuristic like ExM.
We choose to solve it approximately, but rather than running a costly ExM at each iteration of DOGS, we found it empirically more efficient to run a cheaper subroutine.
In line with Section~\ref{sec:PVS_results}, we actually propose to sample from the PVS distribution \eqref{DefPCVS}, over the set $\Xbest\cup\Xnew$, and with reference measure $\nu$ minimizing the continuous relaxation \eqref{eq:Relax}.

\begin{figure}[!ht]
\begin{algorithm}[H]
\SetAlgoLined
\DontPrintSemicolon
1: Start with a random initial design $X_0\in\Omega^k$ and set $\Xbest=X_0$.\;
2: Choose a random design $\Xnew$ in $\Omega^{k'}$ for some $k'\in \N$.\;
3: Let $\{x_1,\cdots,x_\ell\}$, $\ell\leq k+k'$, be the set of unique points in $\Xnew$ and $\Xbest$. Compute a solution $(p_i)_{1\leq i\leq \ell}$ of the convex optimization problem
\begin{equation}\label{opt_in_DOGS}
\hspace{-0.05cm}\argmin~h\left(\sum_{i=1}^{\ell} p_i\phi(x_i)^T\phi(x_i)+\pcov\right)~\mbox{s.t.}~\sum_{i=1}^{\ell}p_i=k~\mbox{and}~p_i\geq 0~\forall i\in\{1,\cdots,k\}.
\end{equation}
4: Sample $Y$ from the $\mathbb{P}_{VS}^\nu(\phi,\pcov)$ distribution over the set $\{x_1,\cdots,x_\ell\}$ where
$$\nu = \sum_{i=1}^\ell p_i\delta_{x_i}.$$
5: If $h(\phi(Y)^T\phi(Y)+\pcov)<h(\phi(\Xbest)^T\phi(\Xbest)+\pcov)$ then set $\Xbest=Y$.\;
6: If some stopping criterion is reached, return $\Xbest$. Otherwise, go back to step 2.\;
 \caption{Discrete optimization within global search}
\end{algorithm}
\caption{{\small The DOGS search heuristic for A- and D-optimal designs.}}
\label{f:DOGS}
\end{figure}

The choice of the proposal distribution from which $\Xnew$ is drawn in Step 2 of Figure~\ref{f:DOGS} is where expert knowledge can be built in the algorithm.
In particular, the proposal can include candidate points $\mathcal{C}$, as explained in Section~\ref{sec:algo_presentation}.
The remaining points can be drawn from any distribution over $\Omega$.
We typically draw i.i.d. samples from the uniform distribution over $\Omega$, since for the design spaces we have found in the literature, the uniform distribution is amenable to rejection sampling.
More space-filling proposals are also possible, like a randomly perturbed grid or a randomly shifted Sobol or Halton sequence \citep{DiPi10}, but we have found that the improvement over the uniform proposal is minor, and the experiments below in Section \ref{sec:perf} thus use i.i.d. uniform draws.
Finally, a large value of the number $k'$ of points in $\Xnew$ makes the algorithm converge in fewer iterations, but each iteration is more costly due to the increase in complexity of the optimization problem \eqref{opt_in_DOGS}.
In short, like for the $\sigma$ parameter of LSA, we recommend tuning $k'$ manually using a few initial short runs.
As for the complexity of the PVS sampling step, we refer to Section \ref{sec:sampling}.

The computational bottleneck of DOGS is the underlying convex optimization subroutine \eqref{opt_in_DOGS}.
The complexity of the latter heavily depends on the optimality criterion and the specific convex optimization algorithm that is used, but it is hard to make a definite scaling argument; we thus show CPU times in Table~\ref{table: Computation_times}.
Finally, unlike PVS in Section \ref{sec:GPCVS}, it is difficult to produce any mathematical result on the output of DOGS. We thus focus on evaluating its empirical performance.

\subsection{Numerical results}
\label{sec:perf}
In this section, we compare the performances of LSA, ExM and DOGS in three applications, of increasing difficulty. The first example comes from \citep[Example 16.3]{Atkinson} and corresponds to a three-component mixture design with quadratic constraints for polynomial regression.
The second example is drawn from \citep{B-Spline_1,B-Spline_2} and illustrates regressing on a functional basis that adds B-Splines to polynomials.
The third example is an engineering application from \citep{Space_Filling_Design}, and corresponds to an eight-component mixture design with a combination of linear and nonlinear constraints for polynomial regression.
We showcase more examples, mostly taken from the textbook \citep{Atkinson}, in a companion Jupyter Notebook provided as supplementary material\footnote{Python code allowing to reproduce all experiments will be publicly released upon acceptance.}. 
We focus on the non-Bayesian setting ($\pcov=0$) since it is more commonly found in the literature. Our personal experience shows that our conclusions on the relative performances of each algorithm in the three examples studied here still hold in the Bayesian setting. 
As it was illustrated in Figure \ref{fig:CompDPP-U}, finding good designs is easier in the Bayesian setting, especially when the prior covariance matrix takes small values.

For all algorithms, initial designs are chosen using i.i.d. uniform draws over $\Omega$.
In LSA, the proposal is generated by adding an i.i.d. Gaussian perturbation to each point in the current design, with a common, manually tuned standard deviation $\sigma$.
If the perturbation of a point makes it leave the design space, we leave the original point as is in the proposed $\Xnew$.
For ExM, the internal optimization subroutine is done using L-BFGS-B in \emph{scipy}~\citep{Scipy}.
In DOGS, the random design $\Xnew$ is sampled uniformly in $\Omega^{k'}$, unless specified otherwise.
The convex optimization in Step $3$ of Figure~\ref{f:DOGS} is carried out using a generic solver from the library \textit{cvxopt}~\citep{cvxopt} for D-optimality, and with the SDP solver of the same library for A-optimality.
Indeed, we follow \cite[Chapter 7.5.2]{Boyd} to cast A-optimality problem \eqref{opt_in_DOGS} as an SDP.
Finally, the PVS distribution in step $4$ of Figure~\ref{f:algo} is sampled using Algorithm \ref{f:algo}, where the DPP distribution is sampled using the \textit{DPPy} library \citep{DPPy}.

\subsubsection{Three-component mixture design with quadratic constraints} \label{sec: Example1}
In \citep[Example 16.3]{Atkinson}, the authors searched for a D-optimal design in
\begin{multline*}
\Omega = \{(x,y)\in\R^2~s.t.~0\leq x+y\leq 1,~-4.062x^2 + 2.962x + y \geq 0.6075\\
\mbox{and}~-1.174x^2 + 1.057x + y \leq 0.5019\}.
\end{multline*}
The regression functions $\phi_i$ are the $6$ multivariate polynomials of $\R^2$ with degree $\leq 2$.
The authors used a finite optimization method over a grid covering the design space, and obtained a six-point approximate design with equal weights.
We thus chose $k=30$, to make their approximate design a true design, i.e. with integer weights, without requiring a rounding step.
We ran a few iterations of DOGS with various values of $k'$, and set for $k'=50$.
Similarly, we set the standard deviation of the Gaussian perturbations in LSA to $\sigma=0.01$.
We present in Figure \ref{fig:perf_mix_deg2} the evolution of the $\log$ D-optimality criterion of the current best design $\Xbest$ across iterations, for $200$ runs of each algorithm. The solid line represents the median $\log$-D-optimality criterion of each algorithm, while the boundary of the shaded area corresponds to the $5$-th and $95$-th percentile. The $\log$ D-optimality criterion of the design in \citep{Atkinson} and the one obtained by the exact algorithm of \cite{DeCastro} are shown as, respectively, a dashed and a dashdotted line.
We also show in Table \ref{table: Computation_times} the runtime of each algorithm on a recent laptop.

We observe that DOGS outperforms both LSA and ExM in this case: DOGS only needs a few hundred iterations, totaling about $10$ seconds, to find a better D-optimal design than the one in \citep{Atkinson}.
In fact,
DOGS settles in the same area of $\Omega$ as the design of \cite{Atkinson}, but is allowed to fine tune its result by not being limited to a grid.
However, this example should now be considered as easy, in the sense that the Lasserre hierarchy of \cite{DeCastro} terminates in about the same time as DOGS and outputs an actual optimal design, with criterion value shown in Figure~\ref{fig:perf_mix_deg2}.

We can easily make the problem harder, though, by considering polynomials of higher degree. On this example, we observed that the Lasserre hierarchy failed to terminate for degrees larger than 2.
This is the regime where search heuristics become useful.
For instance, taking the regression functions $\phi_i$ to be the $15$ multivariate polynomials of $\R^2$ with degree $\leq 4$, we show the results of LSA, ExM, and DOGS in Figure \ref{fig:perf_mix_deg4}.
We also compare these results with the average D-optimality criterion of designs obtained by a discrete exchange method on the set $\Omega\cap 0.01\Z^2$ of $736$ candidate points, to mimic the approach of \citep{Atkinson}.
The results are shown in Figure~\ref{fig:perf_mix_deg4}.
The rankings are similar, with DOGS finding the best solution in a small number of iterations.

\subsubsection{Multifactor B-Spline Mixed Models} \label{sec: Example2}
This example is inspired by \citep{B-Spline_1}. The problem is to find a relationship between several features of a car engine, like maximum brake-torque timing (MBT), and three factors: the engine speed ($x_1$), its air-fuel ratio ($x_2$) and its load ($x_3$).
The authors chose to model the relationship between MBT and $x_2$ and $x_3$ as a cubic polynomial, while the dependency of MBT in the variable $x_1$ is modeled as a maximally smooth cubic B-spline basis with three knots.
We denote this basis by $\{B_1,\cdots,B_7\}$ for the remainder of this paper.
Using all possible products of B-spline and polynomials would mean considering $112$ basis functions of the form
$$\phi:(x_1,x_2,x_3)\mapsto B_i(x_1)x_2^\alpha x_3^\beta,~i\in\{1,\cdots,7\},~\alpha\in\{0,\cdots,3\},~\beta\in\{0,\cdots,3\}.$$
The authors chose to reduce this regression basis to a smaller one of $p=31$ functions detailed in \citep[Equation (6)]{B-Spline_1}.
The size of the optimal design searched is $k=55$.
The design space is unfortunately not specified, but it is alluded that its shape is complicated due to various combinations of $x_1$, $x_2$ and $x_3$ being either unphysical or potentially damaging to the engine.
To illustrate the properties of the different algorithms, we use both a simple and a more complex design space.

We first consider $\Omega=[0,1]^3$ and that the B-Spline knots are located at $0.25$, $0.5$ and $0.75$, although their location did not have any significant impact on the results.
A set of 28 candidate points is suggested in \citep{B-Spline_2}, namely
$$\mathcal{C}\defeq\left\{(x_1,x_2,x_3)\in\Omega,\quad x_1\in\{\argmax B_i, 1\leq i\leq 7\},\quad  x_2,x_3\in\{0,1\}\right\}.$$
To investigate the impact of candidate points on DOGS, we compare a version with a uniform proposal over $\Omega^{50}$, to a version where the proposal is made of the union of $\mathcal{C}$ and a uniform draw over $\Omega^{22}$.
We compare to LSA with $\sigma=0.01$ manually tuned, ExM, and a discrete ExM on the $700$ candidate points
$$\mathcal{C}'\defeq\{(x_1,x_2,x_3)\in\Omega,\quad x_1\in\{\argmax B_i, 1\leq i\leq 7\},\quad x_2,x_3\in\{0,1/9,2/9,\cdots,1\}\}$$
as suggested in \citep{B-Spline_2}.
We present in Figure \ref{fig:perf_BSpline_Cube} the evolution of the $\log$ D-optimality criterion of $\Xbest$ for $200$ runs of each algorithm with respect to their number of iteration and the average $\log$ D-optimality criterion of Discrete ExM. We also present in Table \ref{table: Computation_times} the runtime of each algorithm on a recent laptop.

LSA performs worst, while both versions of ExM tie in converging to what we believe is close to the optimal criterion.
DOGS again yields a quickly decreasing criterion, but its curve plateaus higher than ExM.
Adding candidate points to DOGS lowers the plateau, but still not to the level of ExM.
A tentative argument to explain the success of discrete ExM in this example is the small-dimensional cubic design space, which is particularly amenable to a discretization by a finite grid.
Similarly, for continuous ExM, the shape of $\Omega$ is natively handled by the L-BFGS-B optimization.
An interesting note is that discrete ExM on an agnostic grid of $1000$ candidate points $[0,1]\cap \frac{1}{9}\Z^3$ does not come close to its performance when including candidate points, which confirms the candidate point suggestions of \citep{B-Spline_2}.

In a second variant of the same experiment, we now make the design space more complicated, as the real design space in \citep{B-Spline_2} is suggested to be. We consider
\begin{equation} \label{2Ball}
\Omega'=B((x_c,x_c,x_c),x_c)\cup B((1-x_c,1-x_c,1-x_c),x_c),~\textrm{where}~x_c\defeq \frac{3-\sqrt{3}}{4},
\end{equation}
and where $B(P,R)$ denotes the ball centered at $P$ with radius $R$. This design space corresponds to two tangent Euclidean balls inside the cube $[0,1]^3$, meeting at the point $(1/2,1/2,1/2)$.
We did not include any candidate point in this case, and the proposal in DOGS is uniform over ${\Omega}'^{50}$.
The discrete version of ExM uses the $719$ points of $\Omega'\cap\frac{1}{14}\Z^3$.
The results are shown in Figure \ref{fig:perf_BSpline_2Balls}.

As expected, the performance of both the continuous and discrete version of ExM suffers from the more complex design space, while the DOGS seems more robust to the change of design space, outperforming the best overall result of any other approach in about 100 iterations.
In particular, we can see that continuous ExM shows a huge variance in its results and, even if they don't appear in the quantiles of Figure \ref{fig:perf_BSpline_Cube}, there were a few runs that got stuck in local optima with a criterion larger than $150$, likely due to a bad initial design.
This confirms the recommendation in \citep{Pronzato} to use a restart strategy.
DOGS did not suffer from this kind of issue.


So far, DOGS has proved to be a robust algorithm, that fares particularly well when $\Omega$ has a complex shape or is hard to discretize.
The typical behaviour of DOGS it to quickly lower its criterion before plateauing.
This suggests a hybrid method, where DOGS is used until its improvement in criterion is deemed to be too small, at which point one switches to a local search like LSA.
We leave the investigation of such variants to future work.

\subsubsection{Nepheline crystallization in high-level nuclear waste glass} \label{sec: Example3}
For our last example, we increase the dimensionality to an eight-component mixture design, \cite{Space_Filling_Design} study the propensity of nepheline crystals (NaAlSiO$_4$) to appear during the fabrication of glass made to contain nuclear waste, in relation to the proportion of the $8$ components of the glass. This leads to a $7$-dimensional problem.
Various linear and non-linear constraints make the design space $\Omega$ quite complicated. No regression is mentioned in the paper; the authors rather focus on producing a space-filling design, meaning a general purpose design with points well spread over $\Omega$.
For our experiments, we take the regression functions $\phi_i$ to be the $36$ multivariate polynomials of $\R^7$ with degree $\leq 2$.

We identified $24$ points at the intersection of the linear and non-linear constraints, and use them as candidate points $\mathcal{C}$.
We thus compare DOGS with a uniform proposal over $\Omega^{100}$ to DOGS with a proposal taken as the union between $\mathcal{C}$ and a uniform draw over $\Omega^{76}$.
We still compare to LSA, with $\sigma=0.001$, ExM, and its discrete version on the set $\mathcal{C}\cup(\Omega\cap \frac{1}{3}\Z^7)$ of $309$ points.

Again, we show in Figure \ref{fig:Perf3} the evolution of the $\log$ D-optimality criterion of $\Xbest$ for $200$ runs of each algorithm with respect to iteration number and the average $\log$ D-optimality criterion of Discrete ExM.
We also show in Table \ref{table: Computation_times} the runtime of each algorithm on a recent laptop.
Again, using candidate points significantly improves the performance of DOGS.
We also note that, despite clearly outperforming ExM and LSA, DOGS struggles to get any close to optimality past its initial decrease. In this case, a discrete exchange method on a very rough grid yields significantly better designs.
In our experience, this is a recurrent issue with DOGS in large dimensions.
Using more space-filling proposals, like continuous PVS in \eqref{DefPCVS} or a Sobol sequence \citep{DiPi10} helps a little, but not enough to outperform the discrete ExM.
We conjecture that a combination of DOGS and local moves could take the best of both, but leave this to future work.

\begin{figure}
\centering
\subfloat[Example~\ref{sec: Example1}: degree $\leq 2$.\label{fig:perf_mix_deg2}]{\includegraphics[width=\twofig]{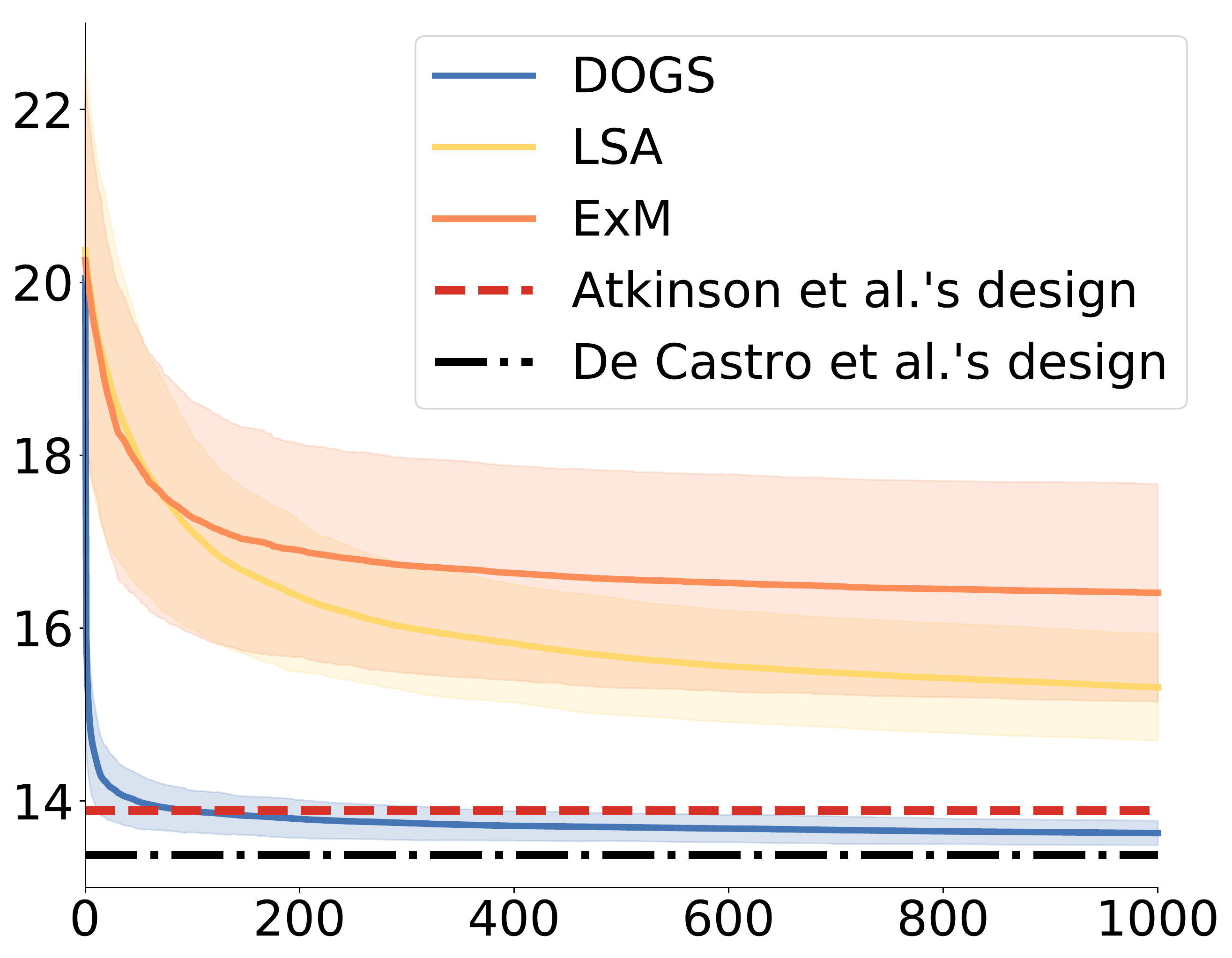}}
\hfill
\subfloat[Example~\ref{sec: Example1}: degree $\leq 4$. \label{fig:perf_mix_deg4}]{\includegraphics[width=\twofig]{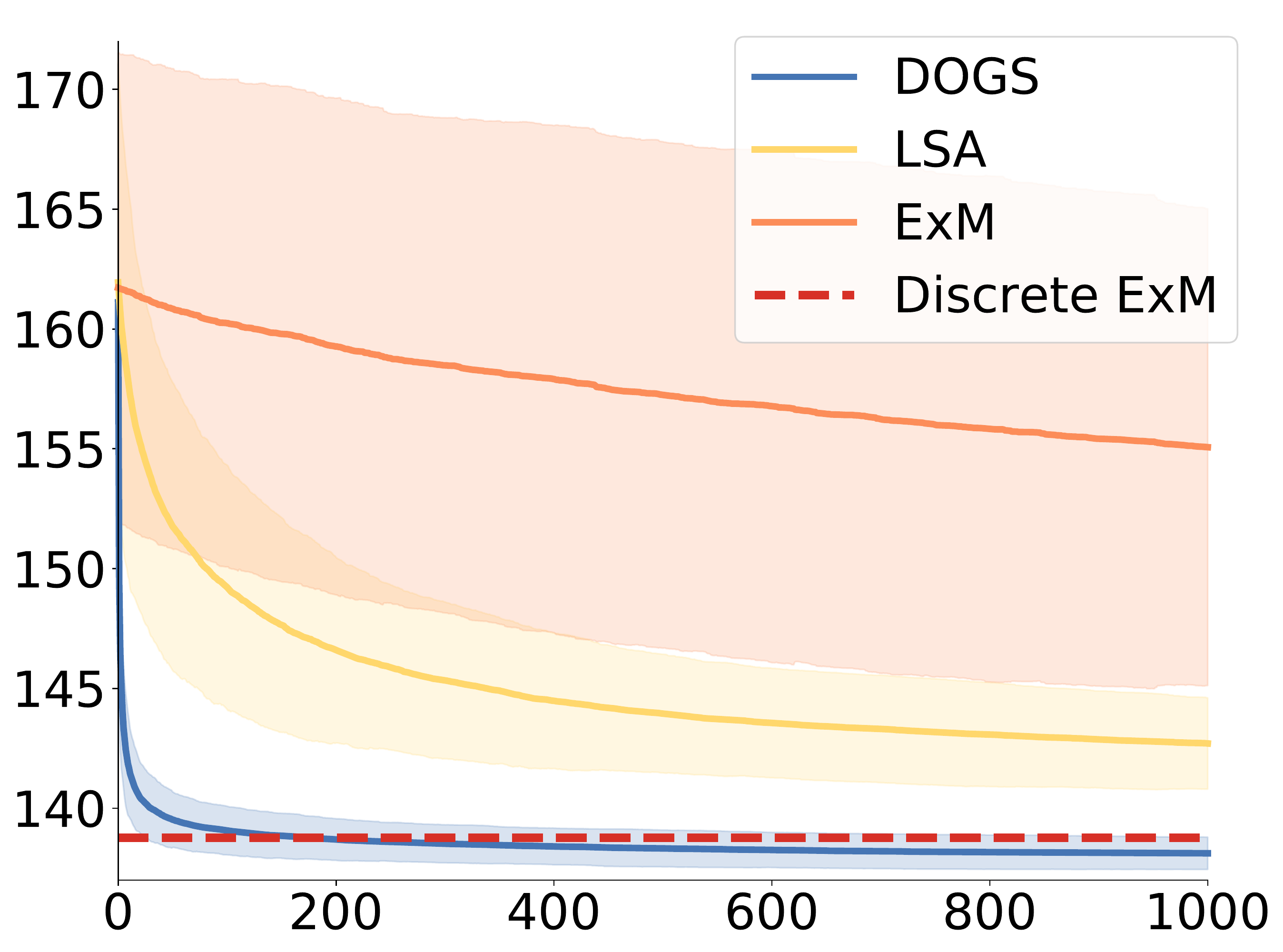}}
\vspace{-0.3cm}
\subfloat[Example~\ref{sec: Example2}: $\Omega={[0,1]}^3$.]{\label{fig:perf_BSpline_Cube}\includegraphics[width=\twofig]{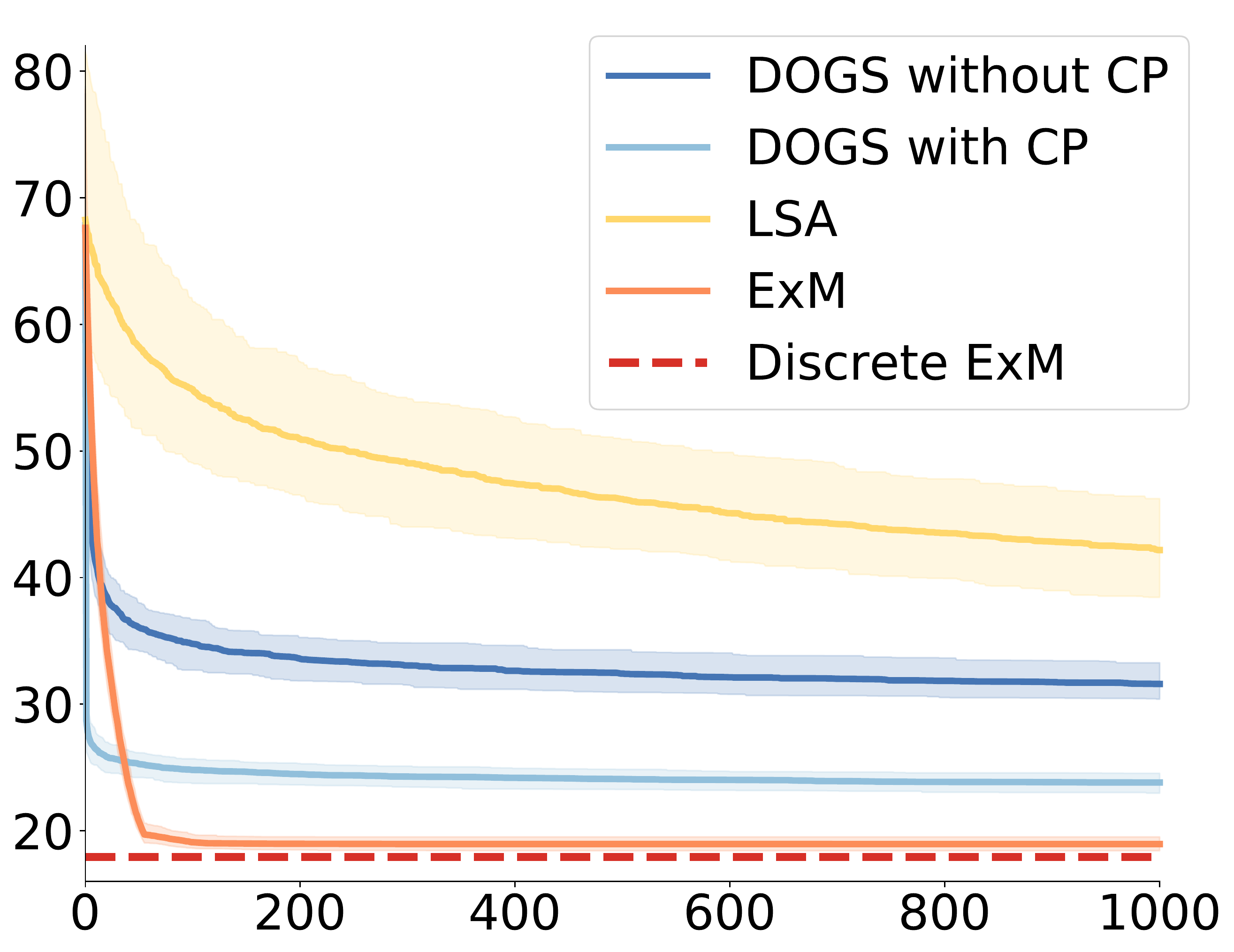}}
\hfill
\subfloat[Example~\ref{sec: Example2}: $\Omega'$ defined in \eqref{2Ball}.]{\label{fig:perf_BSpline_2Balls}\includegraphics[width=\twofig]{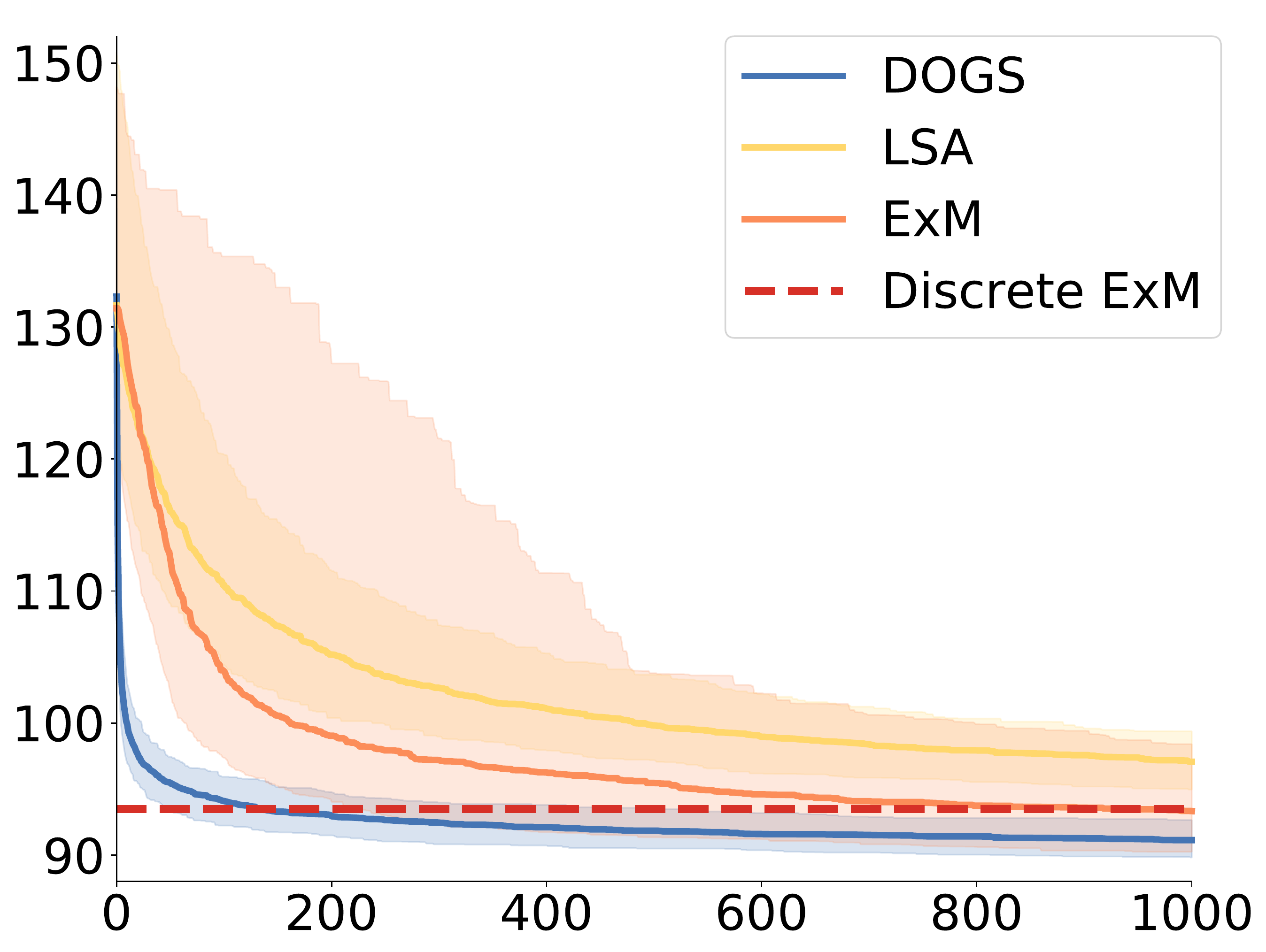}}
\vspace{-0.3cm}
\begin{center}
 \subfloat[Example~\ref{sec: Example3} \label{fig:Perf3}.]{\includegraphics[width=\twofig]{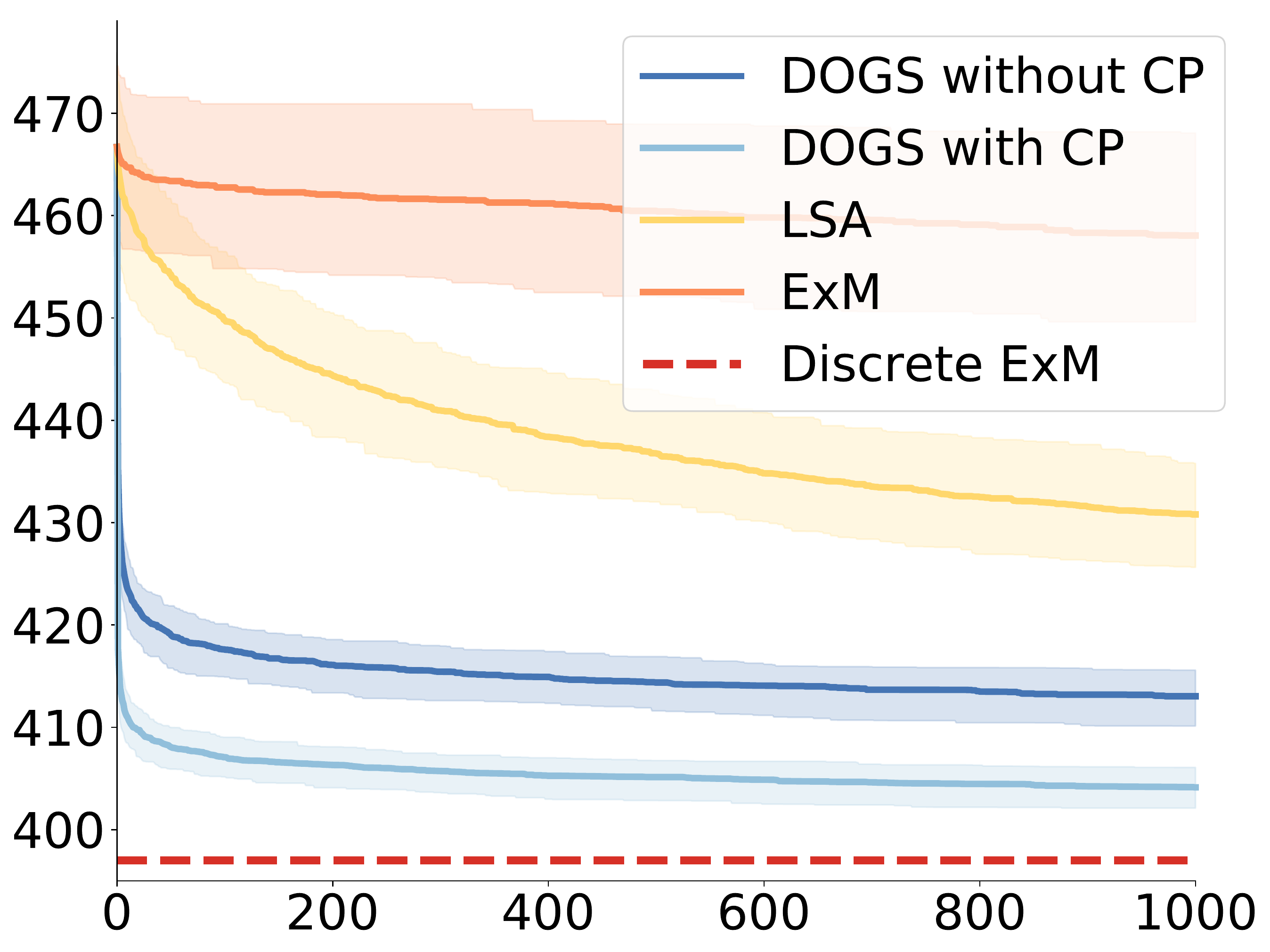}}
\end{center}
\vspace{-0.7cm}
	   \caption{{\small $\log h_D(\phi(\Xbest)^T\phi(\Xbest))$ vs. iteration number, for $200$ runs of each algorithm.
}}
\end{figure}

\begin{table}
\centering
\begin{tabular}{|c|cccc|}
	\hline
	 & DOGS & LSA & ExM & Disc. ExM \\
	\hline
	Example~\ref{sec: Example1}: degree $\leq 2$. & $47.9$s & $0.951$s & $30.2$s &  \\
	Example~\ref{sec: Example1}: degree $\leq 4$. & $51.5$s & $0.940$s & $12.6$s & $17.8$s \\
	Example~\ref{sec: Example2}: $\Omega={[0,1]}^3$. & $121$s & $4.14$s & $32.6$s & $260$s \\
	Example~\ref{sec: Example2}: $\Omega'$ defined in \eqref{2Ball}. & $96.4$s & $2.67$s & $59.6$s & $268$s \\
	Example~\ref{sec: Example3}. & $178$s & $4.63$s & $130$s & $184$s \\
	\hline
\end{tabular}
\caption{Average computation times on a laptop for one realization of discrete ExM and $1000$ iterations of DOGS, LSA and ExM.}
\label{table: Computation_times}
\end{table}

\section{Conclusion}
Our goal was to investigate some of the statistical implications of recent advances on volume sampling for discrete Bayesian optimal design.

We first turned the Bayesian version of finite proportional volume sampling \citep{Nikolov,Def_RescaledVS} introduced by \cite{RegDPP} into a general distribution over any Polish space.
Using point process arguments, we showed that this generalization preserves the property of giving unbiased estimates of the inverse information matrix and its determinant.
Additionally, we proved approximation guarantees for the A-efficiency and D-efficiency of designs sampled from general proportional volume sampling, conditionally to having a fixed size.
Through a connection with determinantal point processes, we highlighted that our general PVS can be sampled in polynomial time.
We also showed that the same algorithm can easily be modified to sample random designs from PVS conditioned on having a fixed size, without using rejection sampling.
This makes PVS a natural tool to extract experimental designs from the solution of the classical convex relaxation \eqref{ODproblem_relax} of the optimal design problem.
However, in spite of its mathematical and methodological support, we found that on simple continuous problems, PVS can be outperformed in practice by simple search heuristics.

We then took a more practical turn, and introduced DOGS, a search heuristic that combines discrete PVS and random sampling to make global moves across a generic design space $\Omega$.
Although it is costlier than popular alternatives, DOGS shines when $\Omega$ has a complicated shape and the dimension remains small, and its behavior is robust to changes in $\Omega$ or the basis functions.
We believe that this makes DOGS a valuable addition to the practitioner's toolbox.
When the ambient dimension $d$ is large ($\geq 5$), however, DOGS fails to find designs as good as a simple discrete exchange method over a reasonable set of candidate points.
This suggests investigating hybrid strategies, e.g., alternating DOGS and local search episodes, or combining search heuristics using, e.g., multi-armed bandits.

\section*{Acknowledgments}
We thank Adrien Hardy for useful discussions throughout the project. We thank Michał Dereziński for his insightful comments and suggestions on an early draft.
We acknowledge support from ERC grant Blackjack (ERC-2019-STG-851866) and ANR AI chair Baccarat (ANR-20-CHIA-0002).

\appendix

\section{Proof of the well-definedness of Definition \ref{def:RDPP}}
\label{proof:RDPP}
It is obvious that the Janossy densities are positive. Therefore, in order to prove that proportional volume sampling is well-defined; see \cite[Proposition 5.3.II.(ii)]{DVJ}, we only need to show that
\begin{equation}\label{toprove}
\sum_{n\geq 0}^\infty\frac{1}{n!}\int_{\Omega^n}j_n(x)\der^n x=1.
\end{equation}
We write the eigenvalues of $\pcov$ as $\lambda_1\leq \cdots\leq \lambda_p$ and the spectral decomposition of $\pcov$ as $\pcov=P^T D_\lambda P$, where $D_\lambda$ is the $p\times p$ diagonal matrix with the $\lambda_i$ as its diagonal entries. Then, we define the functions $\psi_i$, $1\leq i\leq p$, by the linear transform of the function $\phi_i$ defined by $(\psi_1(x),\cdots,\psi_p(x))\defeq (\phi_1(x),\cdots,\phi_p(x))P^T$. Finally, we have the decomposition
\begin{align*}
\det(\phi(x)^T\phi(x)+\pcov)&=\det(P\phi(x)^T\phi(x)P^T+D_\lambda)\\
&=\det(\psi(x)^T\psi(x)+D_\lambda)\\
&=\sum_{S\subset[p]}\lambda^{S^c}\det(\psi_S(x)^T\psi_S(x))
\end{align*}
where $\psi_S\defeq (\psi_{S_1},\cdots,\psi_{S_{|S|}})$ and $\lambda^{S^c}\defeq\prod_{i\notin S}\lambda_i$, with the usual convention $\lambda^{\emptyset}=1$; see \citep{Det_Diag}.
Now, by the discrete Cauchy-Binet formula,
$$\det(\psi_S(x)^T\psi_S(x))=\sum_{\substack{T\subset[k] \\ |T|=|S|}}\det(\psi_S(x_T))^2$$
where $x_T\defeq (x_{T_1},\cdots,x_{T_{|T|}})$. And, using the more general Cauchy-Binet formula \citep{Cauchy_Binet},
$$\int_{\Omega^n}\det(\psi_S(x_T))^2\der\nu^n(x)=|T|!\det(G_\nu(\psi_S))\nu(\Omega)^{n-|T|}.$$
Therefore
\begin{align}
\sum_{n\geq 0}^\infty\frac{1}{n!} \int_{\Omega^n}\det(\phi(x)^T\phi(x)+\pcov)& \der\nu^n(x)\nonumber \\
=&\sum_{n\geq 0}^\infty\frac{1}{n!}\sum_{S\subset[p]}\lambda^{S^c}\int_{\Omega^n}\det(\psi_S(x)^T\psi_S(x))\der\nu^n(x)\nonumber \\
=&\sum_{n\geq 0}^\infty\frac{1}{n!}\sum_{S\subset[p]}\lambda^{S^c}\sum_{\substack{T\subset[k] \\ |T|=|S|}}|T|!\det(G_\nu(\psi_S))\nu(\Omega)^{n-|T|}\nonumber \\
=&\sum_{n\geq 0}^\infty\frac{1}{n!}\sum_{S\subset[p]}\lambda^{S^c}\binom{n}{|S|}|S|!\det(G_\nu(\psi_S))\nu(\Omega)^{n-|S|}\nonumber \\
=&\sum_{n\geq 0}^\infty\sum_{S\subset[p]}\frac{\nu(\Omega)^{n-|S|}}{(n-|S|)!}\lambda^{S^c}\det(G_\nu(\psi_S))\cara{n\geq |S|}\nonumber \\
=&\sum_{S\subset[p]}\lambda^{S^c}\det(G_\nu(\psi_S))\sum_{n\geq |S|}^\infty\frac{\nu(\Omega)^{n-|S|}}{(n-|S|)!}\nonumber \\
=&\det(G_\nu(\psi)+D_\lambda)\exp(\nu(\Omega))\nonumber \\
=&\det(G_\nu(\phi)+\pcov)\exp(\nu(\Omega))\nonumber
\end{align}
where, in the last two identities, we used the facts that $(i)$ $G_\nu(\psi_S)$ is equal to $G_\nu(\psi)_S$, the submatrix of $G_\nu(\psi)$ whose rows and collumns are indexed by $S$, and $(ii)$ $G_\nu(\psi)=G_\nu(\phi P^T)=PG_\nu(\psi)P^T$. This proves \eqref{toprove}.

\section{Proof of Proposition~\ref{prop:Equal_Expectations}}
\label{proof:Equal_Expectations}

First, we write
$$\E\left[(\phi(X)^T\phi(X)+\pcov)^{-1}\right]=\sum_{n\geq 0}\frac{1}{n!}\int_{\Omega^n}(\phi(x)^T\phi(x)+\pcov)^{-1}j_n(x)\der^n x.$$
Since $(\phi(x)^T\phi(x)+\pcov)^{-1}\det(\phi(x)^T\phi(x)+\pcov)$ is the adjugate matrix of $(\phi(x)^T\phi(x)+\pcov)$, its $(i,j)$ entry is $(-1)^{i+j}\det(\phi_{-j}(x)^T\phi_{-i}(x)+\pcov_{-j,-i})$, where we define $\pcov_{-j,-i}$ as the matrix $\pcov$ with its $j$-th row and $i$-th column removed, and $\phi_{-i}(x)$ as the vector $\phi(x)$ with its $i$-th entry removed. Therefore, the $(i,j)$ entry of the matrix $\E\left[(\phi(X)^T\phi(X)+\pcov)^{-1}\right]$ is
$$\sum_{n\geq 0}\frac{1}{n!}\int_{\Omega^n}\frac{(-1)^{i+j}\det(\phi_{-j}(x)^T\phi_{-i}(x)+\pcov_{-j,-i})}{\det(G_\nu(\phi)+\pcov)\exp(\nu(\Omega))}\der\nu^n(x).$$
Using the same reasoning as in the proof of normalization in Section \ref{proof:RDPP}, we get that
\begin{multline}\label{eq:Same_as_Normalisation}
\sum_{n\geq 0}\frac{1}{n!}\int_{\Omega^n}\det(\phi_{-j}(x)^T\phi_{-i}(x)+\pcov_{-j,-i})\der\nu^n(x)\\
=\det\left((\langle \phi_a,\phi_b\rangle)_{\substack{1\leq a,b\leq p \\ a\neq j, b\neq i}}+\pcov_{-j,-i}\right)\exp(\nu(\Omega)).
\end{multline}
Note that the proof in Section \ref{proof:RDPP} does not rely on any symmetricity argument, so that identity \eqref{eq:Same_as_Normalisation} can be proved in the same way. As a consequence we get that
$$\sum_{n\geq 0}\frac{1}{n!}\int_{\Omega^n}\frac{(-1)^{i+j}\det(\phi_{-j}(x)^T\phi_{-i}(x)+\pcov_{-j,-i})}{\det(G_\nu(\phi)+\pcov)\exp(\nu(\Omega))}\der\nu^n(x)=\frac{(-1)^{i+j}\Delta_{j,i}(G_\nu(\phi)+\pcov)}{\det(G_\nu(\phi)+\pcov)},$$
which is the $(i,j)$ entry of the inverse matrix of $G_\nu(\phi)+\pcov$. This proves identity \eqref{Eq_regDPP1}.

Finally, the proof of identity \eqref{Eq_regDPP2} is  straightforward:
\begin{align*}
\E\left[\det(\phi(X)^T\phi(X)+\pcov)^{-1}\right]&=\sum_{n\geq 0}\frac{1}{n!}\int_{\Omega^n}\frac{1}{\det(G_\nu(\phi)+\pcov)\exp(\nu(\Omega))}\der\nu^n(x)\\
&=\frac{\exp(\nu(\Omega))}{\det(G_\nu(\phi)+\pcov)\exp(\nu(\Omega))}\\
&=\det(G_\nu(\phi)+\pcov)^{-1}.
\end{align*}

\section{Proof of Proposition~\ref{prop:regDPP_Dopt}}
\label{proof:regDPP_Dopt}

By definition of the Janossy densities, we have
\begin{equation}\label{eq:step1_Dopt_regDP}
\E\left[\det(\phi(X)^T\phi(X)+\pcov)^{-1}\big| |X|=k\right]=\frac{\frac{1}{k!}\int_{\Omega^k}j_k(x)\det(\phi(x)^T\phi(x)+\pcov)^{-1}\der^k x}{\frac{1}{k!}\int_{\Omega^k}j_k(x)\der^k x}.
\end{equation}
The integral in the numerator simplifies to
\begin{align}
\int_{\Omega^k}j_k(x)\det(\phi(x)^T\phi(x)+\pcov)^{-1}\der^k x&=\int_{\Omega^k}\frac{1}{\det(G_\nu(\phi)+\pcov)\exp(\mu(f))}\der\nu^k(x)\nonumber\\
&=\frac{\nu(\Omega)^k}{\exp(\nu(\Omega))}\det(G_\nu(\phi)+\pcov)^{-1}.\label{eq:step2_Dopt_regDP}
\end{align}
As for the denominator of \eqref{eq:step1_Dopt_regDP}, following the lines of Section \ref{proof:RDPP} leads to
\begin{equation}\label{eq:step3_Dopt_regDP}
\frac{1}{k!}\int_{\Omega^k}j_k(x)\der^k x=\frac{1}{\det(G_\nu(\phi)+\pcov)\exp(\nu(\Omega))}\sum_{S\subset[p]}\lambda^{S^c}\det(G_\nu(\psi_S))\frac{\nu(\Omega)^{k-|S|}}{(k-|S|)!},
\end{equation}
where the $\psi$ functions are defined the same way as in Section \ref{proof:RDPP}.
Recalling that
$$\sum_{S\subset[p]}\lambda^{S^c}\det(G_\nu(\psi_S))=\det(G_\nu(\phi)+\pcov),$$
we can rewrtite the sum in \eqref{eq:step3_Dopt_regDP} as
\begin{multline*}
\sum_{S\subset[p]}\lambda^{S^c}\det(G_\nu(\psi_S))\frac{\nu(\Omega)^{k-|S|}}{(k-|S|)!}\\
=\frac{\nu(\Omega)^{k-p}}{(k-p)!}\det(G_\nu(\phi)+\pcov)+\sum_{\substack{S\subset[p] \\ S\neq [p]}}\lambda^{S^c}\det(G_\nu(\psi_S))\left(\frac{\nu(\Omega)^{k-|S|}}{(k-|S|)!}-\frac{\nu(\Omega)^{k-p}}{(k-p)!}\right).
\end{multline*}
Now, since $\nu(\Omega)=k$, the sequence $i\mapsto \nu(\Omega)^i/i!$ is increasing when $i\leq k$. Hence, for all $S\subset[p]$ such that $S\neq [p]$,
$$\frac{\nu(\Omega)^{k-|S|}}{(k-|S|)!}-\frac{\nu(\Omega)^{k-p}}{(k-p)!}\geq \frac{\nu(\Omega)^{k-p+1}}{(k-p+1)!}-\frac{\nu(\Omega)^{k-p}}{(k-p)!}=\frac{k^{k-p}}{(k-p)!}\times\frac{p-1}{k-p+1}.$$
We thus obtain
\begin{multline}\label{eq:step4_Dopt_regDP}
\sum_{S\subset[p]}\lambda^{S^c}\det(G_\nu(\psi_S))\frac{\nu(\Omega)^{k-|S|}}{(k-|S|)!}\\
\geq\frac{k^{k-p}}{(k-p)!}\left(\det(G_\nu(\phi)+\pcov)+\frac{p-1}{k-p+1}\big(\det(G_\nu(\phi)+\pcov)-\det(G_\nu(\phi))\big)\right).
\end{multline}
Finally, combining \eqref{eq:step1_Dopt_regDP}, \eqref{eq:step2_Dopt_regDP}, \eqref{eq:step4_Dopt_regDP} and the fact that $\nu(\Omega)=k$, we get
\begin{align*}
&\E\left[\det(\phi(X)^T\phi(X)+\pcov)^{-1}\big| |X|=k\right]\\
\leq &\frac{\frac{k^k}{k!\exp(k)}\det(G_\nu(\phi)+\pcov)^{-1}}{\frac{k^{k-p}}{(k-p)!\exp(k)}\left(1+\frac{p-1}{k-p+1}\big(1-\det(G_\nu(\phi)(G_\nu(\phi)+\pcov)^{-1})\big)\right)}\\
= &\frac{k^p(k-p)!}{k!}\left(1+\frac{p-1}{k-p+1}\big(1-\det(G_\nu(\phi)(G_\nu(\phi)+\pcov)^{-1})\big)\right)^{-1}\hspace{-0.3cm}\det(G_\nu(\phi)+\pcov)^{-1},
\end{align*}
concluding the proof.

\section{Proof of Proposition~\ref{prop:PVS_better}}
\label{proof:PVS_better}
Using the convexity of $x\mapsto 1/x$ on $\R_+^*$, it comes
\begin{align*}
\E[\det(\phi(Y)^T\phi(Y)+\pcov)^{-1}]&\geq(\E[\det(\phi(Y)^T\phi(Y)+\pcov)])^{-1}\\
&=\left(\nu(\Omega)^{-k}\int_{\Omega^k} \det(\phi(y)^T\phi(y)+\pcov)\der\nu^k(y)\right)^{-1}.
\end{align*}
Now, in Section \ref{proof:regDPP_Dopt} we showed that
\begin{align*}
\E[\det(\phi(X)^T\phi(X)+\pcov)^{-1}\big| |X|=k]&=\frac{\frac{1}{k!}\int_{\Omega^k}j_k(x)\det(\phi(x)^T\phi(x)+\pcov)^{-1}\der^k x}{\frac{1}{k!}\int_{\Omega^k}j_k(x)\der^k x}\\
&=\frac{\frac{\nu(\Omega)^k}{\exp(\nu(\Omega))}\det(G_\nu(\phi)+\pcov)^{-1}}{\int_{\Omega^k}\frac{\det(\phi(x)^T\phi(x)+\pcov)}{\det(G_\nu(\phi)+\pcov)\exp(\nu(\Omega))}\der\nu^k(x)}\\
&=\nu(\Omega)^k\left(\int_{\Omega^k}\det(\phi(x)^T\phi(x)+\pcov)\der\nu^k(x)\right)^{-1}
\end{align*}
which concludes the proof.

\section{Proof of Proposition~\ref{prop:regDPP_Aopt}}
\label{proof:regDPP_Aopt}

By definition of the Janossy densities, we have
\begin{equation}\label{eq:step1_Aopt_regDP}
\E\left[\tr((\phi(X)^T\phi(X)+\pcov)^{-1})\big| |X|=k\right]=\frac{\frac{1}{k!}\int_{\Omega^k}j_k(x)\tr((\phi(x)^T\phi(x)+\pcov)^{-1})\der^k x}{\frac{1}{k!}\int_{\Omega^k}j_k(x)\der^k x}.
\end{equation}
Using the same notation as in Section \ref{proof:RDPP}, we expand the numerator into
\begin{align}
\frac{1}{k!}\int_{\Omega^k}j_k(x) \tr((\phi(x)^T\phi(x)+\pcov)^{-1}) & \der^k x\nonumber\\
&=\frac{1}{k!}\int_{\Omega^k}j_k(x)\tr((\psi(x)^T\psi(x)+D_\lambda)^{-1})\der^k x\nonumber\\
&= \sum_{i=1}^p\int_{\Omega^n}\frac{\Delta_{i,i}(\psi(x)^T\psi(x)+D_{\lambda})}{\det(G_\nu(\phi)+\pcov)\exp(\nu(\Omega))}\der\nu^k(x).\label{eq:inter2_Aopt_regDP}
\end{align}
Now,
\begin{equation}\label{eq:inter_Aopt_regDP}
\int_{\Omega^n}\Delta_{i,i}(\psi(x)^T\psi(x)+D_{\lambda})\der\nu^k(x)=\sum_{S\subset[p]\backslash\{i\}}\lambda^{S^c}\det(G_\nu(\psi_S))\frac{\nu(\Omega)^{k-|S|}}{(k-|S|)!},
\end{equation}
where in this case $S^c$ denotes the complement of $S$ relative to $[p]\backslash\{i\}$. Note that there are exactly $\dim(\textrm{Ker}(\pcov))$ eigenvalues of $\pcov$ equal to $0$, so that the elements in the sum in \eqref{eq:inter_Aopt_regDP} are equal to $0$ when $|S|\leq m_0-1$. Since $\nu(\Omega)=k$, the sequence $i\mapsto \nu(\Omega)^i/i!$ is increasing when $i\leq k$, so that
\begin{align*}
\int_{\Omega^n}\Delta_{i,i}(\psi(x)^T\psi(x)+D_{\lambda})\der\nu^k(x)&\leq\sum_{S\subset[p]\backslash\{i\}}\lambda^{S^c}\det(G_\nu(\psi_S))\frac{\nu(\Omega)^{k+1-m_0}}{(k+1-m_0)!}\nonumber\\
&=\frac{\nu(\Omega)^{k+1-m_0}}{(k+1-m_0)!}\Delta_{i,i}(G_\nu(\psi)+D_{\lambda})
\end{align*}
which, combined with \eqref{eq:inter2_Aopt_regDP}, gives
\begin{equation}\label{eq:step2_Aopt_regDP}
\frac{1}{k!}\int_{\Omega^k}j_k(x)\tr((\phi(x)^T\phi(x)+\pcov)^{-1})\der^k x\leq\frac{\nu(\Omega)^{k+1-m_0}\tr((G_\nu(\phi)+\pcov)^{-1})}{(k+1-m_0)!\exp(\nu(\Omega))}.
\end{equation}
Finally, combining \eqref{eq:step1_Aopt_regDP}, \eqref{eq:step2_Aopt_regDP} and \eqref{eq:step3_Dopt_regDP} gives
\begin{align*}
\E\left[\det(\phi(X)^T\phi(X)+\pcov)^{-1}\big| |X|=k\right]&\leq\frac{\frac{\nu(\Omega)^{k+1-m_0}\tr((G_\nu(\phi)+\pcov)^{-1})}{(k+1-m_0)!\exp(\nu(\Omega))}}{\frac{\nu(\Omega)^{k-p}}{(k-p)!\exp(\nu(\Omega))}}\\
&=\frac{\nu(\Omega)^{p+1-m_0}(k-p)!}{(k+1-m_0)!}\tr\big((G_\nu(\phi)+\pcov)^{-1}\big)
\end{align*}
and since $\nu(\Omega)=k$, this concludes the proof.

\section{Proof of Proposition~\ref{prop:Corr_functions}}
\label{proof:Corr_functions}

The Janossy densities and correlation functions of a point process are linked by the following identity; see \cite[Lemma 5.4.III]{DVJ}:
$$\rho_n(x_1,\cdots,x_n)\der^n x=\sum_{m\geq 0}\frac{1}{m!}\int_{\Omega^m}j_{n+m}(x,y)\der^m y.$$
Applying this identity to the Janossy densities of $\PVS{\nu}{\phi}{\pcov}$, we get
$$\rho_n(x_1,\cdots,x_n)\der^n x=\sum_{m\geq 0}\frac{1}{m!}\int_{\Omega^m}\frac{\det(\phi(x)^T\phi(x)+\phi(y)^T\phi(y)+\pcov)}{\det(G_\nu(\phi)+\pcov)\exp(\nu(\Omega))}\der\nu^m(y).$$
Now, for all $x_1,\cdots,x_n\in\Omega$, using the same reasoning as in the proof of normalization in Section \ref{proof:RDPP} but replacing the matrix $\pcov$ with the matrix $\phi(x)^T\phi(x)+\pcov$, we get
$$\sum_{m\geq 0}\frac{1}{m!}\int_{\Omega^m}\det(\phi(x)^T\phi(x)+\phi(y)^T\phi(y)+\pcov)\der\nu^m(y)=\det(\phi(x)^T\phi(x)+\pcov)\exp(\nu(\Omega)).$$
We then conclude that
$$\rho_n(x_1,\cdots,x_n)=\frac{\det(G_\nu(\phi)+\pcov+\phi(x)^T\phi(x))}{\det(G_\nu(\phi)+\pcov)}\prod_{i=1}^nf(x_i).$$

\section{Proof of Proposition~\ref{prop:Superpos}}
\label{proof:Superpos}

For any $n\in\N$ and $x\in\Omega^n$, we write $K[x]$ for the $n\times n$ matrix with entries $K(x_i,x_j)$. Since $G_\nu(\phi)+\pcov$ is invertible, then
\begin{align}
\rho_n(x)\der^n x&=\det(I_p+(G_\nu(\phi)+\pcov)^{-1}\phi(x)^T\phi(x))\der\nu^n(x)\nonumber\\
&=\det(I_n+\phi(x)(G_\nu(\phi)+\pcov)^{-1}\phi(x)^T)\der\nu^n(x).\label{eq:Nicest_intensity_expression}\\
&=\det(I_n+\phi(x)G_\nu(\phi)^{-1/2}PDP^TG_\nu(\phi)^{-1/2}\phi(x)^T)\der\nu^n(x)\nonumber\\
&=\det(I_n+\psi(x)D\psi(x)^T)\der\nu^n(x)\nonumber\\
&=\det(I_n+K[x])\der\nu^n(x).\nonumber
\end{align}
Now, it remains to show that the superposition of $X$ and $Y$ has the same correlation functions to conclude that its distribution is $\PVS{\nu}{\phi}{A}$.


Let $n\in\N$, we recall that the $n$-th order correlation function $\rho'_n$ of $X\cup Y$ satisfy
\begin{equation}\label{eq:Expect}
\E\left[\sum_{x_1,\cdots,x_n\in X\cup Y}^{\neq}f(x_1,\cdots,x_n)\right]=\int_{\Omega^n}f(x_1,\cdots,x_n)\rho'_n(x_1,\cdots,x_n)\der x_1\cdots \der x_n
\end{equation}
for all integrable functions $f$, where the $\neq$ symbol means that the sum is taken on distinct elements of $X\cup Y$. Since each element of $X\cup Y$ is either in $X$ or $Y$ then \eqref{eq:Expect} can be rewritten as
\begin{align*}
\E\left[\sum_{x_1,\cdots,x_n\in X\cup Y}^{\neq}f(x_1,\cdots,x_n)\right]&=\sum_{S\subset[n]}\E\left[\sum_{x_i\in X, i\in S}^{\neq}\sum_{x_j\in Y, j\in S^c}^{\neq}f(x_1,\cdots,x_n)\right]\\
&=\sum_{S\subset[n]}\E\left[\sum_{x_i\in X, i\in S}^{\neq}\E\left[\sum_{x_j\in Y, j\in S^c}^{\neq}f(x_1,\cdots,x_n)\right]\right]\\
&=\sum_{S\subset[n]}\E\left[\sum_{x_i\in X, i\in S}^{\neq}\int_{\Omega^{|S^c|}}f(x_1,\cdots,x_n)\prod_{j\in S^C}\der\nu(x_j)\right]\\
&=\sum_{S\subset[n]}\int_{\Omega^n}f(x_1,\cdots,x_n)\det((K(x_i,x_j))_{i,j\in S})\der\nu^n(x)\\
&=\int_{\Omega^n}f(x_1,\cdots,x_n)\det(I_n+K[x])\der\nu^n(x).
\end{align*}
This proves that the correlation functions of $X\cup Y$ also satisfy
$$\rho'_n(x)\der^n x=\det(I_n+K[x])\der\nu^n(x_i).$$
Therefore, $X\cup Y$ is distributed as $\PVS{\nu}{\phi}{A}$.

\section{Proof of Corollary~\ref{corr:Avg_num_points}}
\label{proof:Avg_num_points}

$X\sim \PVS{\nu}{\phi}{\pcov}$ is the superposition of a Poisson point process $Y$ with intensity $\nu$ and a DPP $Z$ with intensity $\rho(x)\der x=\phi(x)(G_\nu(\phi)+\pcov)^{-1}\phi(x)^T$; see identity \eqref{eq:Nicest_intensity_expression}. Therefore,
$$\E[|X|]=\E[|Y|]+\E[|Z|]$$
with $\E[|Y|]=\nu(\Omega)$ and
$$\E[|Z|]=\int_\Omega \phi(x)(G_\nu(\phi)+\pcov)^{-1}\phi(x)^T\der\nu(x).$$
Since we can rewrite $\phi(x)(G_\nu(\phi)+\pcov)^{-1}\phi(x)^T$ as $\tr((G_\nu(\phi)+\pcov)^{-1}\phi(x)^T\phi(x))$, we get
$$\E[|Z|]=\tr((G_\nu(\phi)+\pcov)^{-1}G_\nu(\phi)),$$
concluding the proof.

\section{A parametrized reference measure for Section~\ref{s:simple_experiment}.}
\label{experiment:approximation_OD}

To parametrize $\nu$, we write its density $f$ as a linear combination of positive functions with nonnegative weights, that is,
\begin{equation} \label{fsum}
f(x)=\sum_{i=1}^n\omega_i g_i(x).
\end{equation}
This way, minimizing $h(G_\nu(\phi))$ over $\nu=f\mathrm{d}x$ of the form \eqref{fsum} and such that $\nu(\Omega)=k$ is equivalent to finding $(\omega_1,\cdots,\omega_n)$ minimizing
\begin{equation}\label{eq:Relax}
h\left(\sum_{i=1}^n \omega_iG_{g_i}(\phi) + \pcov\right)~\mbox{s.t.}~\omega\succcurlyeq 0~\mbox{and}~\sum_{i=1}^n\omega_i\int_\Omega g_i(x)\der x=k.
\end{equation}
This is now a convex optimization problem that can be solved numerically.
For our illustration, we consider that $h\in\{h_D,h_A\}$ and the $g_i$ to be the $231$ polynomial functions of two variables with degree $\leq 10$ as well as their composition with $(x,y)\mapsto (1-x,1-y)$, which are all non-negative functions on $\Omega=[0,1]^2$.
We show in Figure \ref{fig:Opti_Poly_Dens} the density of the measures minimizing \eqref{eq:Relax} for both optimality criteria and for $\pcov\in\{I_{10}, 0.01 I_{10}, 0.0001 I_{10}\}$.

\begin{figure}[H]
\centering
   \subfloat[D-optimality, $\pcov=I_d$\label{fig:Opti_Poly_Dens_D_1}]{\includegraphics[width=\densplot]{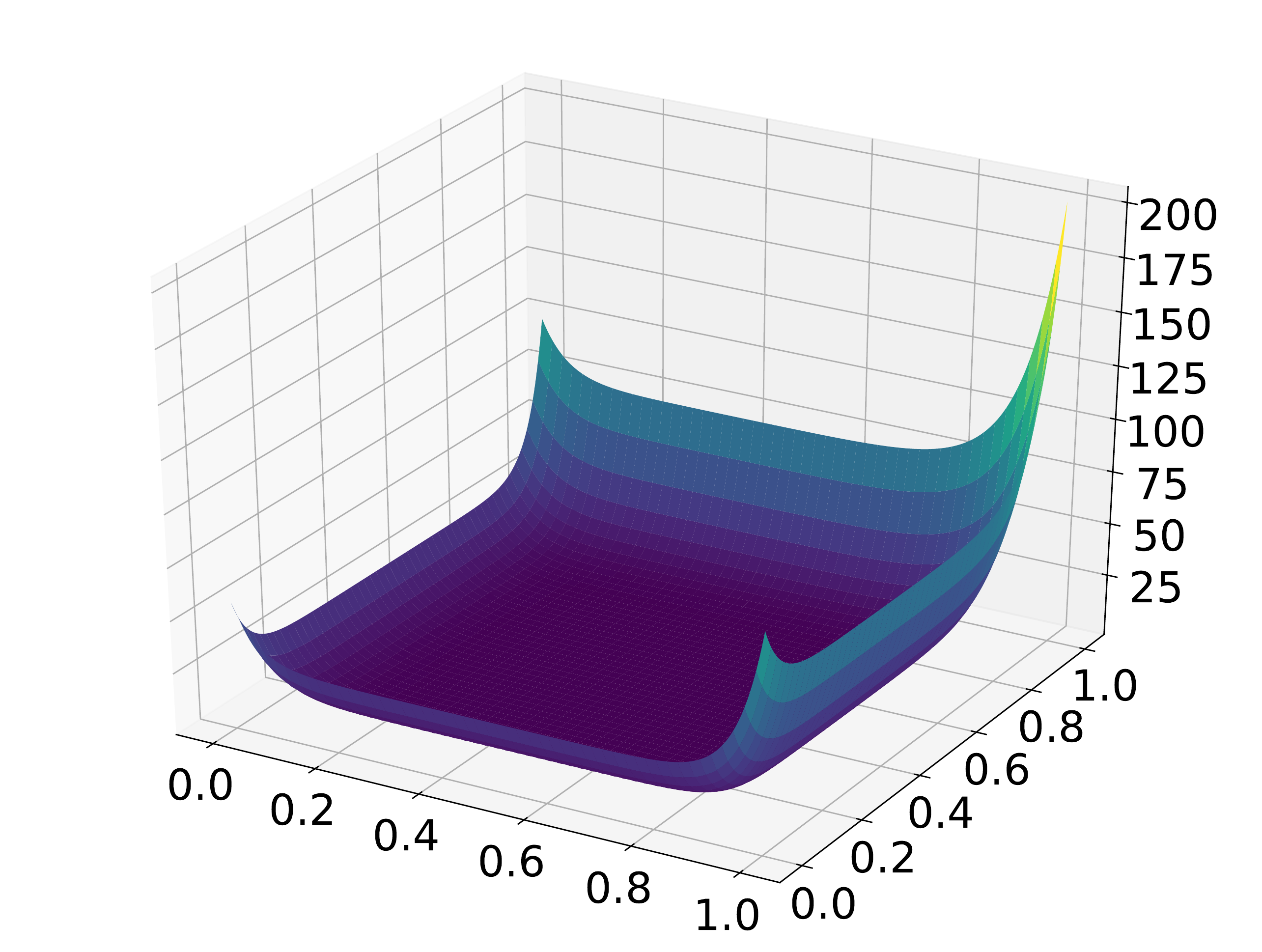}}\hfill
      \subfloat[D-optimality, $\pcov=0.01 I_d$\label{fig:Opti_Poly_Dens_D_100}]{\includegraphics[width=\densplot]{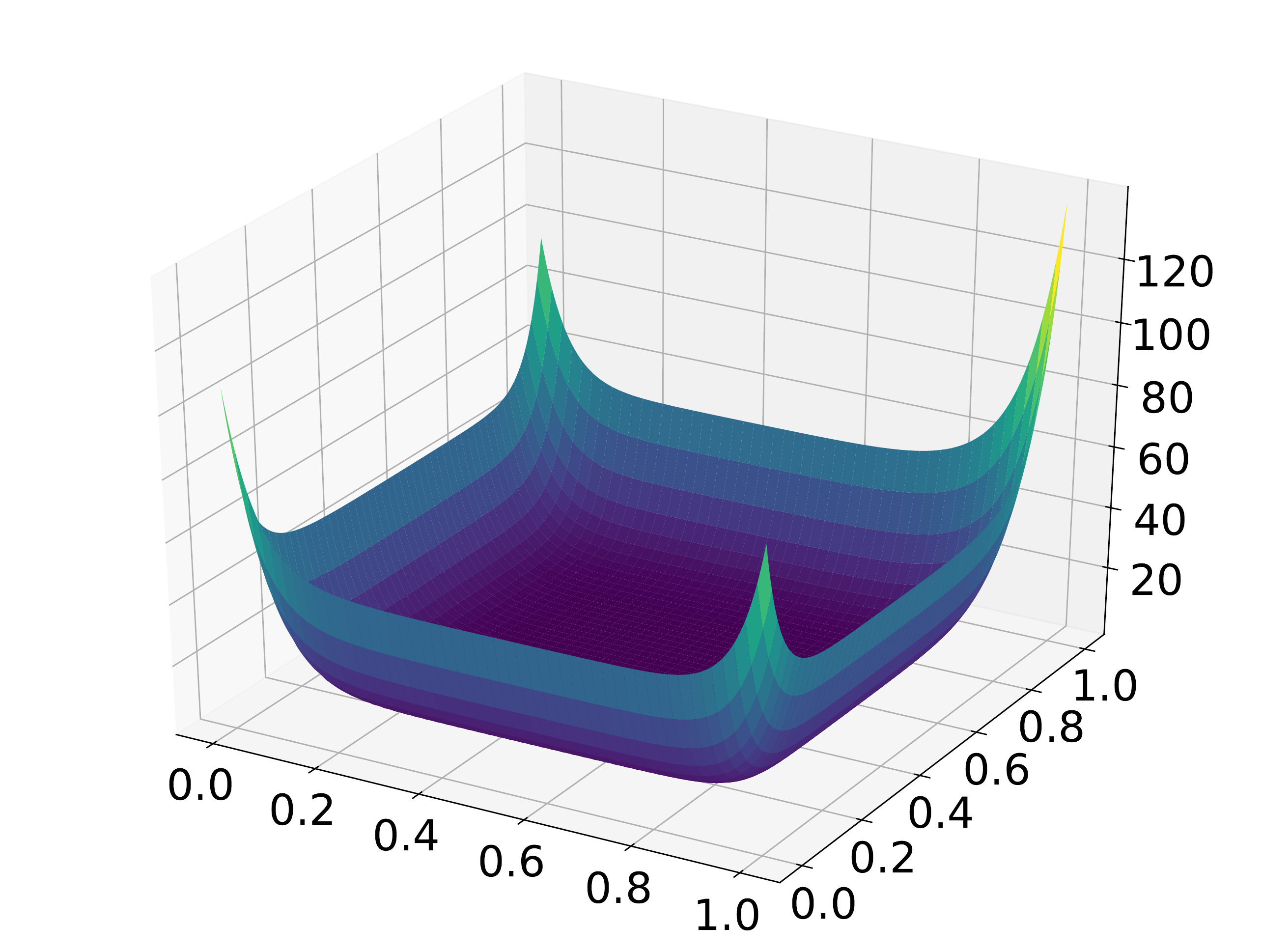}}\hfill
         \subfloat[D-optimality, $\pcov=0.0001 I_d$\label{fig:Opti_Poly_Dens_D_10000}]{\includegraphics[width=\densplot]{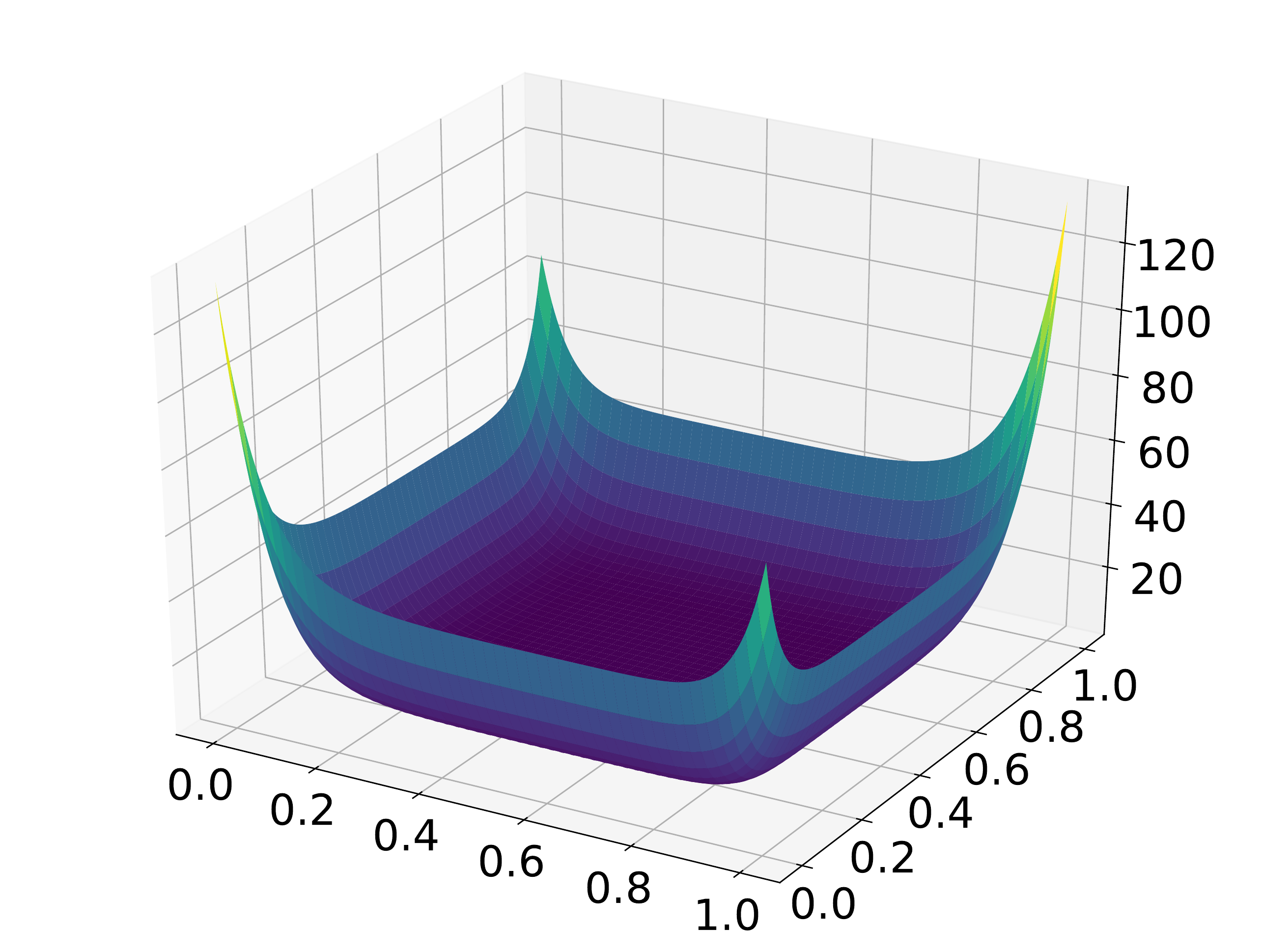}}\hfill
\subfloat[A-optimality, $\pcov=I_d$\label{fig:Opti_Poly_Dens_A_1}]{\includegraphics[width=\densplot]{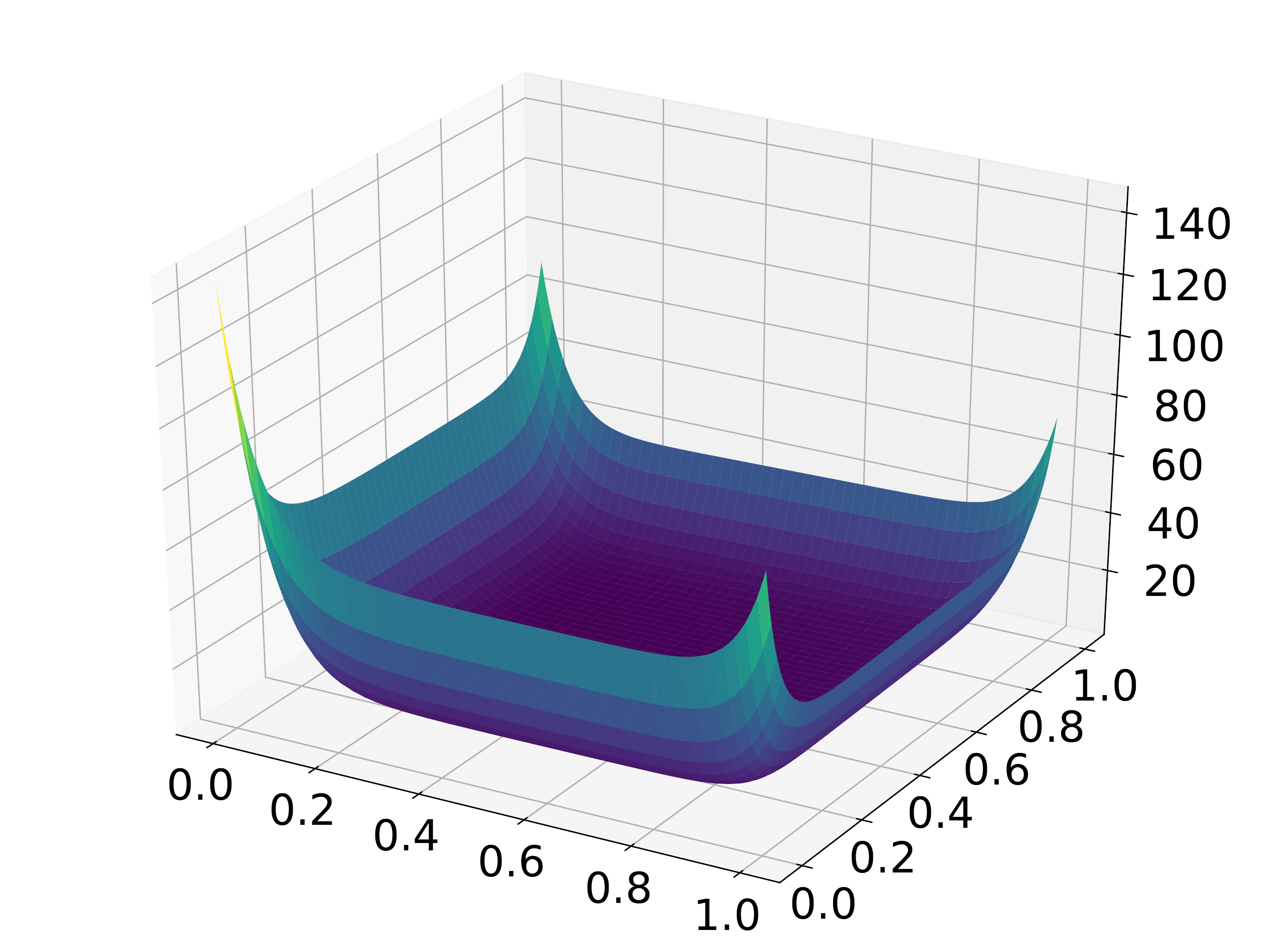}}\hfill
\subfloat[A-optimality, $\pcov=0.01 I_d$\label{fig:Opti_Poly_Dens_A_100}]{\includegraphics[width=\densplot]{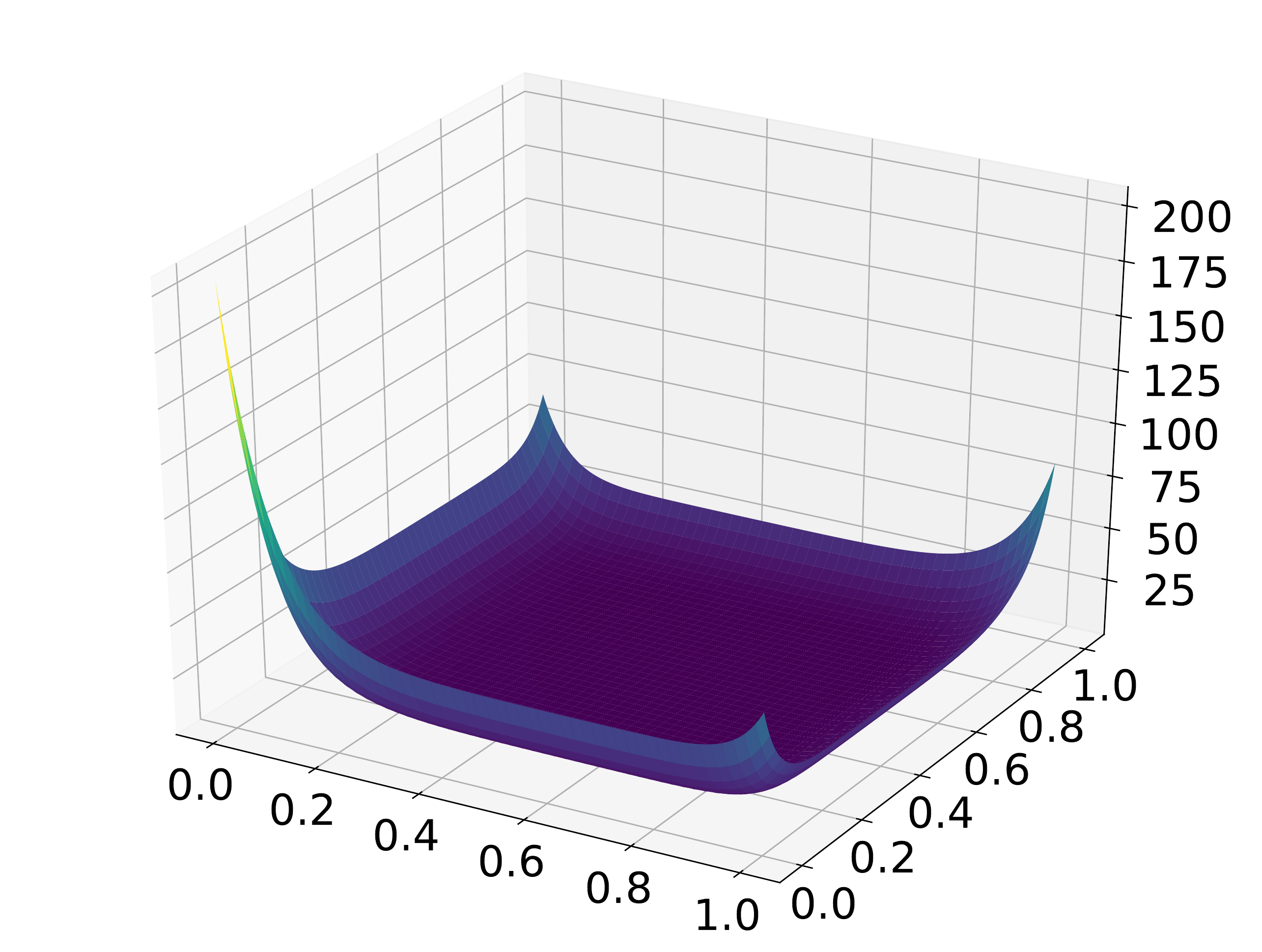}}\hfill
\subfloat[A-optimality, $\pcov=0.0001 I_d$\label{fig:Opti_Poly_Dens_A_10000}]{\includegraphics[width=\densplot]{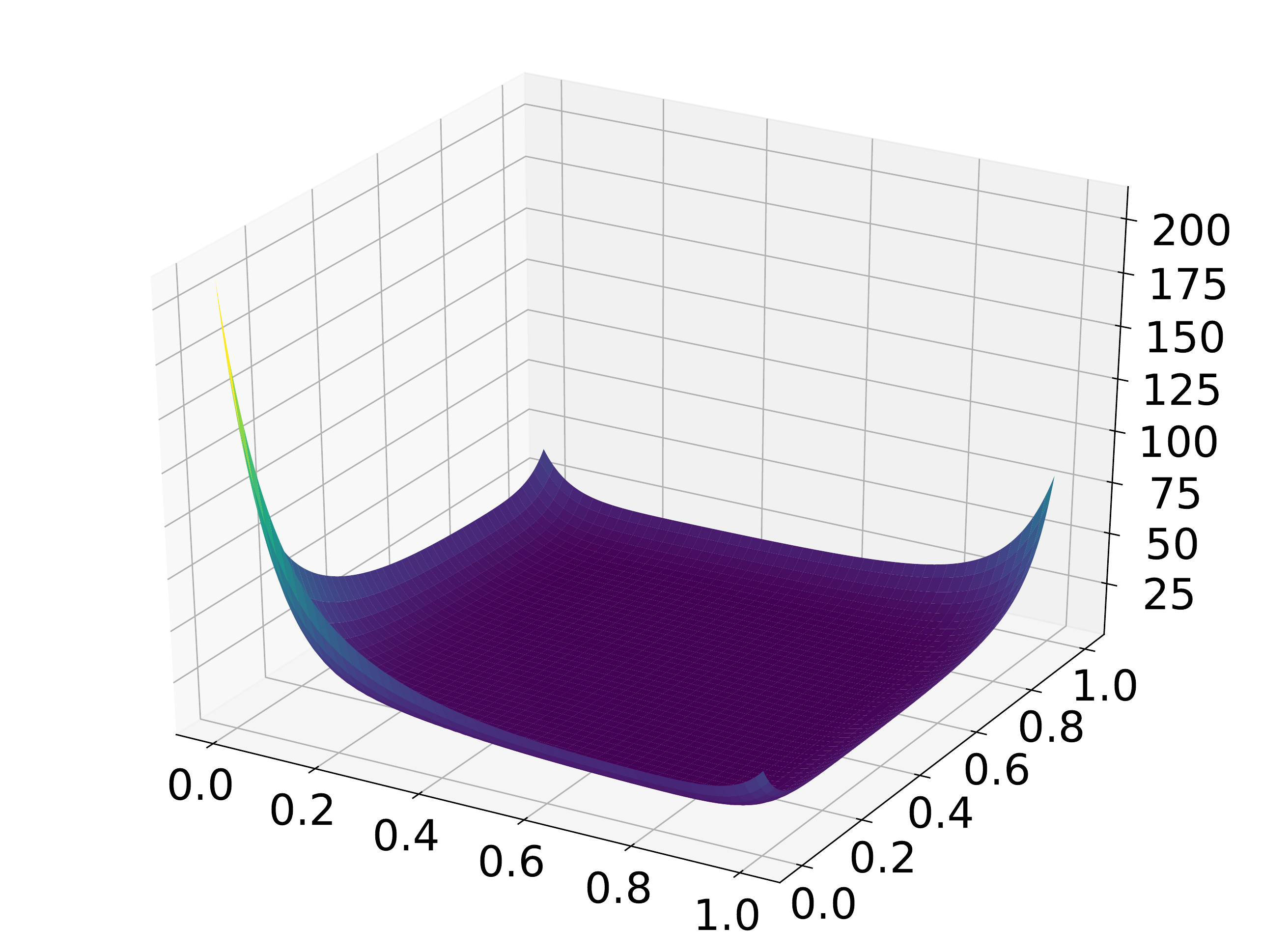}}
	   \caption{{\label{fig:Opti_Poly_Dens} \small 3D plots of the densities of the measures minimizing \eqref{eq:Relax} for the D and A-optimality criterion when the $g_i$ functions are the binomial polynomial of degree $\leq 10$ as well as their composition with $(x,y)\mapsto (1-x,1-y)$.}}
\end{figure}

\bibliographystyle{abbrvnat}
\bibliography{ref,remi}

\begin{thebibliography}{37}
\providecommand{\natexlab}[1]{#1}
\providecommand{\url}[1]{\texttt{#1}}
\expandafter\ifx\csname urlstyle\endcsname\relax
  \providecommand{\doi}[1]{doi: #1}\else
  \providecommand{\doi}{doi: \begingroup \urlstyle{rm}\Url}\fi

\bibitem[Andersen et~al.(2012)Andersen, Dahl, Liu, and Vandenberghe]{cvxopt}
M.~Andersen, J.~Dahl, Z.~Liu, and L.~Vandenberghe.
\newblock Interior-point methods for large-scale cone programming.
\newblock In S.~Sra, S.~Nowozin, and S.~Wright, editors, \emph{Optimization for
  Machine Learning}, chapter~1, pages 55--83. MIT Press, 2012.

\bibitem[Atkinson et~al.(2007)Atkinson, Donev, and Tobias]{Atkinson}
A.~Atkinson, A.~Donev, and R.~Tobias.
\newblock \emph{Optimum Experimental Designs, with SAS}.
\newblock Oxford Statistical Science Series. Oxford University Press, USA,
  2007.

\bibitem[Boyd and Vandenberghe(2004)]{Boyd}
S.~Boyd and L.~Vandenberghe.
\newblock \emph{Convex Optimization}.
\newblock Cambridge University Press, USA, 2004.

\bibitem[Collings(1983)]{Det_Diag}
B.~J. Collings.
\newblock Characteristic polynomials by diagonal expansion.
\newblock \emph{The American Statistician}, 37\penalty0 (3):\penalty0 233--235,
  1983.

\bibitem[Daley and Vere-Jones(2003)]{DVJ}
D.~J. Daley and D.~Vere-Jones.
\newblock \emph{An Introduction to the Theory of Point Processes. {V}ol. {I}}.
\newblock Springer-Verlag, 2nd edition, 2003.

\bibitem[De~Castro et~al.(2019)De~Castro, Gamboa, Henrion, Hess, and
  Lasserre]{DeCastro}
Y.~De~Castro, F.~Gamboa, D.~Henrion, R.~Hess, and J.~Lasserre.
\newblock Approximate optimal designs for multivariate polynomial regression.
\newblock \emph{Annals of Statistics}, 47\penalty0 (1):\penalty0 127--155, 02
  2019.

\bibitem[Derezi\'{n}ski et~al.(2018)Derezi\'{n}ski, Warmuth, and
  Hsu]{Def_RescaledVS}
M.~Derezi\'{n}ski, M.~Warmuth, and D.~Hsu.
\newblock Leveraged volume sampling for linear regression.
\newblock In \emph{Advances in Neural Information Processing Systems 31: Annual
  Conference on Neural Information Processing Systems}, pages 2510--2519, 2018.

\bibitem[Derezi\'{n}ski et~al.(2019)Derezi\'{n}ski, Warmuth, and Hsu]{VRS}
M.~Derezi\'{n}ski, M.~Warmuth, and D.~Hsu.
\newblock Unbiased estimators for random design regression, 2019.
\newblock arXiv pre-print.

\bibitem[Derezi\'{n}ski et~al.(2020)Derezi\'{n}ski, Liang, and Mahoney]{RegDPP}
M.~Derezi\'{n}ski, F.~Liang, and M.~Mahoney.
\newblock Bayesian experimental design using regularized determinantal point
  processes.
\newblock In S.~Chiappa and R.~Calandra, editors, \emph{Proceedings of the
  Twenty Third International Conference on Artificial Intelligence and
  Statistics}, volume 108 of \emph{Proceedings of Machine Learning Research},
  pages 3197--3207, Online, 26--28 Aug 2020. PMLR.

\bibitem[Dette(1993)]{Dette_Bayes}
H.~Dette.
\newblock Bayesian d-optimal and model robust designs in linear regression
  models.
\newblock \emph{Statistics: A Journal of Theoretical and Applied Statistics},
  25\penalty0 (1):\penalty0 27--46, 1993.

\bibitem[Dette and Studden(1997)]{Dette}
H.~Dette and W.~J. Studden.
\newblock \emph{The Theory of Canonical Moments with Applications in
  Statistics, Probability, and Analysis}.
\newblock Wiley Series in Probability and Statistics. Wiley, 1997.
\newblock ISBN 9780471109914.

\bibitem[Dette et~al.(2002)Dette, Melas, and Pepelyshev]{Trigo_regression}
H.~Dette, V.~Melas, and A.~Pepelyshev.
\newblock D-optimal designs for trigonometric regression models on a partial
  circle.
\newblock \emph{Annals of the Institute of Statistical Mathematics},
  54:\penalty0 945--959, 02 2002.

\bibitem[Dick and Pilichshammer(2010)]{DiPi10}
J.~Dick and F.~Pilichshammer.
\newblock \emph{Digital Nets and Sequences. Discrepancy Theory and Quasi-Monte
  Carlo Integration}.
\newblock Cambridge University Press, 2010.

\bibitem[Fang et~al.(2006)Fang, Li, and Sudjianto]{LocSearch}
K.~Fang, R.~Li, and A.~Sudjianto.
\newblock \emph{Design and modeling for computer experiments}.
\newblock Computer science and data analysis series. Chapman and Hall/CRC, 1
  edition, 2006.

\bibitem[Farrell et~al.(1967)Farrell, Kiefer, and Walbran]{Multi_Designs}
R.~H. Farrell, J.~Kiefer, and A.~Walbran.
\newblock Optimum multivariate designs.
\newblock In \emph{Proceedings of the Fifth Berkeley Symposium on Mathematical
  Statistics and Probability, Volume 1: Statistics}, pages 113--138. University
  of California Press, 1967.

\bibitem[Fedorov(1972)]{Fedorov}
V.~Fedorov.
\newblock \emph{Theory of Optimal Experiments Designs}.
\newblock Academic Press, New York, 01 1972.

\bibitem[Gautier et~al.(2019{\natexlab{a}})Gautier, Bardenet, and
  Valko]{GaBaVa19a}
G.~Gautier, R.~Bardenet, and M.~Valko.
\newblock On two ways to use determinantal point processes for {M}onte {C}arlo
  integration.
\newblock Technical report, ICML workshop on Negative dependence in machine
  learning, 2019{\natexlab{a}}.

\bibitem[Gautier et~al.(2019{\natexlab{b}})Gautier, Polito, Bardenet, and
  Valko]{DPPy}
G.~Gautier, G.~Polito, R.~Bardenet, and M.~Valko.
\newblock {DPPy: DPP Sampling with Python}.
\newblock \emph{Journal of Machine Learning Research - Machine Learning Open
  Source Software (JMLR-MLOSS)}, 2019{\natexlab{b}}.

\bibitem[Grove et~al.(2004)Grove, Woods, and Lewis]{B-Spline_1}
D.~Grove, D.~Woods, and S.~Lewis.
\newblock Multifactor b-spline mixed models in designed experiments for the
  engine mapping problem.
\newblock \emph{Journal of Quality Technology}, 36\penalty0 (4):\penalty0
  380--391, 2004.

\bibitem[Hough et~al.(2009)Hough, Krishnapur, Peres, and Virag]{Hough}
J.~Hough, M.~Krishnapur, Y.~Peres, and B.~Virag.
\newblock \emph{Zeros of Gaussian Analytic Functions and Determinantal Point
  Processes}.
\newblock American Mathematical Society, 2009.
\newblock ISBN 978-0-8218-43.

\bibitem[Hough et~al.(2006)Hough, Krishnapur, Peres, and Vir\'ag]{HKPV06}
J.~B. Hough, M.~Krishnapur, Y.~Peres, and B.~Vir\'ag.
\newblock Determinantal processes and independence.
\newblock \emph{Probability surveys}, 2006.

\bibitem[Johansson(2006)]{Cauchy_Binet}
K.~Johansson.
\newblock Random matrices and determinantal processes.
\newblock In \emph{Les Houches Summer School Proceedings}, volume 83(C), pages
  1--56, 2006.

\bibitem[Kulesza and Taskar(2012)]{KuTa12}
A.~Kulesza and B.~Taskar.
\newblock Determinantal point processes for machine learning.
\newblock \emph{Foundations and Trends in Machine Learning}, 2012.

\bibitem[Lavancier et~al.(2015)Lavancier, M{\o}ller, and Rubak]{Lavancier}
F.~Lavancier, J.~M{\o}ller, and E.~Rubak.
\newblock Determinantal point process models and statistical inference.
\newblock \emph{Journal of Royal Statistical Society: Series B (Statistical
  Methodology)}, 77:\penalty0 853--877, 5 2015.

\bibitem[Liski et~al.(2002)Liski, Mandal, Shah, and Sinha]{Liski}
E.~Liski, N.~Mandal, K.~Shah, and B.~a. Sinha.
\newblock \emph{Topics in Optimal Design}.
\newblock Lecture Notes in Statistics 163. Springer-Verlag New York, 1 edition,
  2002.
\newblock ISBN 978-0-387-95348-9,978-1-4613-0049-6.

\bibitem[Liu et~al.(2020)Liu, Yue, and Chatterjee]{Bayes_Example}
X.~Liu, R.-X. Yue, and K.~Chatterjee.
\newblock Geometric characterization of d-optimal designs for random
  coefficient regression models.
\newblock \emph{Statistics and Probability Letters}, 159:\penalty0 108696,
  2020.

\bibitem[Macchi(1975)]{Mac75}
O.~Macchi.
\newblock The coincidence approach to stochastic point processes.
\newblock \emph{Advances in Applied Probability}, 7:\penalty0 83--122, 1975.

\bibitem[Maronge et~al.(2017)Maronge, Zhai, Wiens, and Fang]{Wavelet_reg}
J.~Maronge, Y.~Zhai, D.~Wiens, and Z.~Fang.
\newblock Optimal designs for spline wavelet regression models.
\newblock \emph{Journal of Statistical Planning and Inference}, 184:\penalty0
  94 -- 104, 2017.

\bibitem[Nikolov et~al.(2019)Nikolov, Singh, and Tantipongpipat]{Nikolov}
A.~Nikolov, M.~Singh, and U.~T. Tantipongpipat.
\newblock Proportional volume sampling and approximation algorithms for
  a-optimal design.
\newblock In \emph{Proceedings of the Thirtieth Annual ACM-SIAM Symposium on
  Discrete Algorithms}, SODA ’19, page 1369–1386. Society for Industrial
  and Applied Mathematics, 2019.

\bibitem[Piepel et~al.(2019)Piepel, Stanfill, Cooley, Jones, Kroll, and
  Vienna]{Space_Filling_Design}
G.~Piepel, B.~Stanfill, S.~Cooley, B.~Jones, J.~Kroll, and J.~Vienna.
\newblock Developing a space-filling mixture experiment design when the
  components are subject to linear and nonlinear constraints.
\newblock \emph{Quality Engineering}, 31\penalty0 (3):\penalty0 463--472, 2019.
\newblock \doi{10.1080/08982112.2018.1517887}.

\bibitem[Pronzato and Pázman(2013)]{Pronzato}
L.~Pronzato and A.~Pázman.
\newblock \emph{Design of Experiments in Nonlinear Models: Asymptotic
  Normality, Optimality Criteria and Small-Sample Properties}.
\newblock Lecture Notes in Statistics 212. Springer-Verlag New York, 2013.

\bibitem[Pukelsheim(2006)]{Pukelsheim}
F.~Pukelsheim.
\newblock \emph{Optimal Design of Experiments}.
\newblock Classics in applied mathematics 50. Society for Industrial and
  Applied Mathematics, 2006.

\bibitem[Pukelsheim and Rieder(1992)]{Efficient_rounding}
F.~Pukelsheim and S.~Rieder.
\newblock Efficient rounding of approximate designs.
\newblock \emph{Biometrika}, 79\penalty0 (4):\penalty0 763--770, 12 1992.

\bibitem[Robert and Casella(2004)]{RoCa04}
C.~P. Robert and G.~Casella.
\newblock \emph{Monte {C}arlo statistical methods}.
\newblock Springer, 2004.

\bibitem[Summa et~al.(2014)Summa, Eisenbrand, Faenza, and
  Moldenhauer]{D_opt_NP_Hard}
M.~Summa, F.~Eisenbrand, Y.~Faenza, and C.~Moldenhauer.
\newblock On largest volume simplices and sub-determinants.
\newblock \emph{Proceedings of the Annual ACM-SIAM Symposium on Discrete
  Algorithms}, 2015, 06 2014.

\bibitem[Virtanen et~al.(2020)Virtanen, Gommers, Oliphant, et~al.]{Scipy}
P.~Virtanen, R.~Gommers, T.~Oliphant, et~al.
\newblock {{SciPy} 1.0: Fundamental Algorithms for Scientific Computing in
  Python}.
\newblock \emph{Nature Methods}, 17:\penalty0 261--272, 2020.

\bibitem[Woods et~al.(2003)Woods, Lewis, and Dewynne]{B-Spline_2}
D.~Woods, S.~Lewis, and J.~Dewynne.
\newblock Designing experiments for multi-variable b-spline models.
\newblock \emph{Sankhya}, 65:\penalty0 660--670, 2003.

\end{thebibliography}
\end{document}